%% file: main.tex
\documentclass[submission,copyright,creativecommons]{eptcs}
\providecommand{\event}{ACT 2021} 
\usepackage{breakurl}             
\usepackage{underscore}           
\usepackage{amsfonts, amsthm, amsmath}
\usepackage{mathtools}
\usepackage{cleveref}
\usepackage{stmaryrd}
\usepackage{enumerate}

\usepackage{graphicx}
\graphicspath{ {./handwritten/cropped/handwritten-} }

\usepackage[svgnames]{xcolor}

\definecolor[named]{pawcolor}{cmyk}{0.80,0,1,0.19} 
\definecolor[named]{gbcolor}{cmyk}{1,0.1,0,0.2} 
\definecolor[named]{asidecolor}{gray}{0.7} 

\newcommand{\IH}{\mathbb{IH}}

\newcommand \LinRel[1]{{\mathsf{LinRel}}_{\scriptscriptstyle #1}}
\newcommand{\AffRel}[1]{\mathsf{AffRel}_{\scriptstyle #1}}
\newcommand{\myeq}[1]{\mathrel{\overset{\makebox[0pt]{\mbox{\normalfont\scriptsize\sffamily (#1)}}}{=}}}
\newcommand{\semElec}[1]{\mathcal{I}\left(#1\right)}
\newcommand{\semEq}{\mathrel{\overset{\mathcal{I}}{=}}}
\newcommand{\semLeq}{\mathrel{\overset{\mathcal{I}}{\leq}}}
\newcommand{\Defeq}{\mathrel{\overset{\text{def}}{=}}}
\newcommand{\sort}[2]{\ensuremath{(#1,\,#2)}}
\newcommand{\typ}{\mathrel{:}}
\newcommand{\bluebullet}{\color{blue}\bullet\color{black}}
\newcommand{\from}{\mathrel{:}}

\theoremstyle{definition}
\newtheorem{theorem}{Theorem}
\newtheorem{lemma}[theorem]{Lemma}
\newtheorem{proposition}[theorem]{Proposition}
\newtheorem{observation}[theorem]{Observation}

\newtheorem{corollary}[theorem]{Corollary}

\usepackage{tikz}
\usetikzlibrary{arrows,backgrounds,shapes.geometric,shapes.misc,matrix,arrows.meta,decorations.markings}

\pgfdeclarelayer{background}
\pgfdeclarelayer{edgelayer}
\pgfdeclarelayer{nodelayer}
\pgfsetlayers{background,edgelayer,nodelayer,main}
\tikzset{baseline=-0.5ex}
\definecolor{light-gray}{gray}{.7}
\tikzstyle{none}=[inner sep=0pt]
\tikzstyle{plain}=[inner sep=0pt]
\tikzstyle{black}=[circle, draw=black, fill=black, inner sep=0pt, minimum size=3.5pt]
\tikzstyle{black-faded}=[circle, draw=light-gray, fill=light-gray, inner sep=0pt, minimum size=4pt]
\tikzstyle{white}=[circle, draw=black, fill=white, inner sep=0pt, minimum size=3.5pt]
\tikzstyle{white-faded}=[circle, draw=light-gray, fill=white, inner sep=0pt, minimum size=4.5pt]
\tikzstyle{delay}=[fill=black, regular polygon, regular polygon sides=3,rotate=-90, scale=.55]
\tikzstyle{delay-op}=[fill=black, regular polygon, regular polygon sides=3,rotate=90, scale=.55]
\tikzstyle{reg}=[draw, fill=white, rounded rectangle, rounded rectangle left arc=none, minimum height=1em, minimum width=1em, node font={\scriptsize}]
\tikzstyle{coreg}=[draw, fill=white, rounded rectangle, rounded rectangle right arc=none, minimum height=1em, minimum width=1em, node font={\scriptsize}]
\tikzstyle{rn}=[circle, draw=red, fill=red, inner sep=0pt, minimum size=4pt]
\tikzstyle{place}=[circle, draw=black, fill=white, inner sep=0pt, minimum size=8pt]

\tikzstyle{medium box}=[fill=white, draw=black, shape=rectangle, minimum height=1cm, minimum width=0.75cm]
\tikzstyle{small box}=[fill=white, draw=black, shape=rectangle, minimum height=0.75cm, minimum width=0.5cm]
\tikzstyle{square}=[fill=white, draw=black, shape=rectangle, minimum height=1cm, minimum width=1cm]
\tikzstyle{triangle}=[fill=black, draw=black, shape=regular polygon, regular polygon sides=3, rotate=270, scale=0.6]
\tikzstyle{antipode}=[fill=black, draw=black, shape=rectangle]
\tikzstyle{triangleop}=[fill=black, draw=black, shape=regular polygon, regular polygon sides=3, rotate=90, scale=0.6]
\tikzstyle{label}=[fill=none, draw=none, shape=circle]
\tikzstyle{transition}=[fill=white, draw=black, shape=rectangle, minimum height=0.75cm, minimum width=0.75cm]

\tikzstyle{transition}=[rectangle,thick,draw=black!75,fill=black!20,minimum size=7pt]
\tikzstyle{place}=[circle, draw=black, fill=white, inner sep=0pt, minimum size=8pt]
\tikzstyle{arrow}=[->]

\newcommand{\Wunit}{
\tikzset{x=1em, y=2.1ex}
\InputIfFileExists{./generators/zero.tikz}{}{\input{./tikz/./generators/zero.tikz}}
\tikzset{x=1em, y=1.5ex}
}
\newcommand{\Bcounit}{
\tikzset{x=1em, y=2.1ex}
\InputIfFileExists{./generators/delete.tikz}{}{\input{./tikz/./generators/delete.tikz}}
\tikzset{x=1em, y=1.5ex}
}
\newcommand{\Bcomult}{
\tikzset{x=1em, y=2.1ex}
\InputIfFileExists{./generators/copy.tikz}{}{\input{./tikz/./generators/copy.tikz}}
\tikzset{x=1em, y=1.5ex}
}
\newcommand{\Wmult}{
\tikzset{x=1em, y=2.1ex}
\InputIfFileExists{./generators/add.tikz}{}{\input{./tikz/./generators/add.tikz}}
\tikzset{x=1em, y=1.5ex}
}
\newcommand{\scalar}{
\tikzset{x=1em, y=2.1ex}
\InputIfFileExists{./generators/scalar.tikz}{}{\input{./tikz/./generators/scalar.tikz}}
\tikzset{x=1em, y=1.5ex}
}
\newcommand{\circuitX}{
\tikzset{x=1em, y=2.1ex}
\InputIfFileExists{./generators/register.tikz}{}{\input{./tikz/./generators/register.tikz}}
\tikzset{x=1em, y=1.5ex}
}
\newcommand{\Bunit}{
\tikzset{x=1em, y=2.1ex}
\InputIfFileExists{./generators/co-delete.tikz}{}{\input{./tikz/./generators/co-delete.tikz}}
\tikzset{x=1em, y=1.5ex}
}
\newcommand{\Wcounit}{
\tikzset{x=1em, y=2.1ex}
\InputIfFileExists{./generators/co-zero.tikz}{}{\input{./tikz/./generators/co-zero.tikz}}
\tikzset{x=1em, y=1.5ex}
}
\newcommand{\Wcomult}{
\tikzset{x=1em, y=2.1ex}
\InputIfFileExists{./generators/co-add.tikz}{}{\input{./tikz/./generators/co-add.tikz}}
\tikzset{x=1em, y=1.5ex}
}
\newcommand{\Bmult}{
\tikzset{x=1em, y=2.1ex}
\InputIfFileExists{./generators/co-copy.tikz}{}{\input{./tikz/./generators/co-copy.tikz}}
\tikzset{x=1em, y=1.5ex}
}
\newcommand{\scalarop}{
\tikzset{x=1em, y=2.1ex}
\InputIfFileExists{./generators/co-scalar.tikz}{}{\input{./tikz/./generators/co-scalar.tikz}}
\tikzset{x=1em, y=1.5ex}
}
\newcommand{\circuitXop}{
\tikzset{x=1em, y=2.1ex}
\InputIfFileExists{./generators/co-register.tikz}{}{\input{./tikz/./generators/co-register.tikz}}
\tikzset{x=1em, y=1.5ex}
}
\newcommand{\One}{
\tikzset{x=1em, y=2.1ex}
\InputIfFileExists{one.tikz}{}{\input{./tikz/one.tikz}}
\tikzset{x=1em, y=1.5ex}
}
\newcommand{\R}{\mathbb{R}}

\newcommand{\antipode}{\lower0pt\hbox{$\includegraphics[height=.2cm]{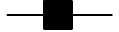}$}}
\newcommand{\impedanceBox}{\lower10pt\hbox{$\includegraphics[height=.8cm]{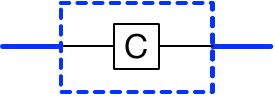}$}}
\newcommand{\wire}{\lower0pt\hbox{$\includegraphics[height=.07cm]{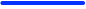}$}}
\newcommand{\opencircuit}{\lower0pt\hbox{$\includegraphics[height=.12cm]{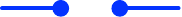}$}}
\newcommand{\resistor}{\lower5pt\hbox{$\includegraphics[height=.8cm]{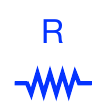}$}}
\newcommand{\vsource}{\lower5pt\hbox{$\includegraphics[height=.8cm]{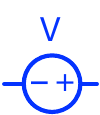}$}}
\newcommand{\csource}{\lower5pt\hbox{$\includegraphics[height=.8cm]{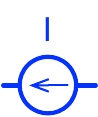}$}}
\newcommand{\inductor}{\lower5pt\hbox{$\includegraphics[height=.8cm]{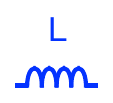}$}}
\newcommand{\capacitor}{\lower5pt\hbox{$\includegraphics[height=.8cm]{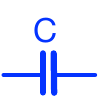}$}}
\newcommand{\junction}{\lower10pt\hbox{$\includegraphics[height=.8cm]{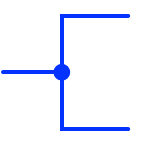}$}}
\newcommand{\circcounit}{\includegraphics[height=.12cm]{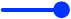}}
\newcommand{\cojunction}{\lower10pt\hbox{$\includegraphics[height=.8cm]{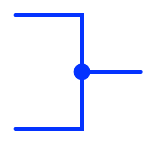}$}}
\newcommand{\ampmeter}{\lower5.5pt\hbox{$\includegraphics[height=.7cm]{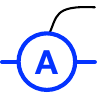}$}}
\newcommand{\voltmeter}{\lower5.5pt\hbox{$\includegraphics[height=.7cm]{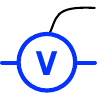}$}}
\newcommand{\cvsource}{\lower5.5pt\hbox{$\includegraphics[height=.7cm]{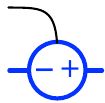}$}}
\newcommand{\ccsource}{\lower5.5pt\hbox{$\includegraphics[height=.7cm]{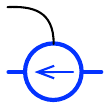}$}}
\newcommand{\circunit}{\includegraphics[height=.12cm]{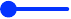}}

\newcommand{\circprop}{\mathsf{ECirc}}
\newcommand{\ecircprop}{\mathsf{EECirc}}

\title{String Diagrammatic Electrical Circuit Theory}
\def\titlerunning{String Diagrammatic Electrical Circuit Theory}

\author{Guillaume Boisseau
\institute{University of Oxford, UK\thanks{Funded by the EPSRC}}
\and
Pawe{\l} Soboci\'nski
\institute{Tallinn University of Technology, Estonia\thanks{Supported by the ESF funded Estonian
IT Academy research measure (project 2014-2020.4.05.19-0001) and the Estonian Research Council grant PRG1210.}}
}
\def\authorrunning{G. Boisseau \& P. Soboci\'nski}

\begin{document}
\maketitle

\begin{abstract}
    We develop a comprehensive string diagrammatic treatment of electrical circuits.
    Building on previous, limited case studies, we introduce controlled sources and meters as elements,
    and the \emph{impedance calculus}, a powerful toolbox for diagrammatic reasoning on circuit
    diagrams. We demonstrate the power of our approach by giving idiomatic proofs of several
    textbook results, including the superposition theorem and Th\'evenin's theorem.
\end{abstract}

\section{Introduction}

Classical electrical circuit theory~\cite{desoerBasicCircuitTheory1969}, and the ubiquitous use of electrical circuits in the daily life of an engineer, are testament to their power as a compositional, diagrammatic abstraction, denoting actual physical systems while at the same time intuitively reflecting network interconnections. Like other formalisms (e.g.\ signal flow graphs), theorems about circuits are proved by translating them---sometimes not entirely faithfully---to more traditional mathematics, with linear algebra playing a particularly important role, and compositionality taking a back seat. Nevertheless, \emph{diagrammatic reasoning} at the level of circuits is the remit of several classical theorems, e.g.\ Th\'evenin's theorem or Norton's theorem.

We build on recent work~\cite{baezCompositionalFrameworkPassive2015,coyaCircuitsBondGraphs2018,bonchiGraphicalAffineAlgebra2019}
that treats circuits compositionally. In particular, we use Coya's semantics~\cite{coyaCircuitsBondGraphs2018}
that captures non-passive elements such as sources, which was subsequently axiomatised in the string-diagrammatic setting of Graphical Affine Algebra (GAA)~\cite{bonchiGraphicalAffineAlgebra2019}, also building on the categorical approach to relational algebra developed by Carboni and Walters~\cite{carboniCartesianBicategories1987}.

This is our starting point. Staying with affine relations as our ambient denotational universe, we enlarge the set of basic electrical elements with \emph{controlled sources} and \emph{meters}: voltmeters and ammeters. This lets us study both open and \emph{closed} circuits; the latter have been ignored in existing compositional approaches since a closed circuit has only two possible denotations as an affine relation.

Our second innovation is what we call the \emph{impedance calculus}. GAA and circuit syntax are intermingled via the use of \emph{impedance boxes}: circuit elements that hold within them a GAA diagram. They reveal the symmetries of classical circuit elements more clearly than ordinary GAA diagrams, but they are more than just of aesthetic value. Indeed, they allow us to solve another common problem of string diagrammatic approaches: having to deal with large diagrams. Using the algebra of impedance boxes allows us to ``compile'' parts of the circuit to GAA, simplify, and translate back into circuit syntax, dramatically reducing the complexity of diagrammatic reasoning when solving typical textbook circuit problems.

Armed with our richer set of circuits and the impedance calculus, we are able to prove most of the well-known textbook (e.g.\ \cite{desoerBasicCircuitTheory1969}) theorems of elementary circuit theory. For space reasons we have not been able to include many; nevertheless, we hope that the assortment on offer will give the reader a taste of the power of diagrammatic reasoning with the impedance calculus. 
On the one hand, the diagrammatic syntax allows for rigorous yet elegant statements of results about circuits, often revealing the underlying symmetries---e.g.\ the independent measurement principle (Theorem~\ref{thm:im}) vs the Superposition Theorem (Theorem~\ref{thm:superposition}); the principle of relativity of potentials (Proposition~\ref{prop:relativePotentials}) vs the principle of conservation of currents (Proposition~\ref{prop:conservationCurrents}). On the other, we show that the compositional framework---with its use of diagrammatic reasoning and the algebra of cartesian and abelian categories of relations---leads to elegant and  rigorous proofs.

\section{Graphical Affine Algebra and Electrical Circuits}


We begin with the basics of Graphical Linear Algebra (GLA) with its equational theory of Interacting Hopf Algebras ($\IH$)~\cite{bonchiInteractingHopfAlgebras2017}. Fix a field $\mathsf{k}$.
GLA is a string diagrammatic syntax, which organises itself as the arrows of the free prop $\mathsf{GLA}_{\mathsf{k}}$ over the following monoidal signature:
\begin{align}
    \{\;&\Bcounit,\ \Bcomult,\  \scalar,\ \Wmult,\ \Wunit,\  \label{eq:GLAsyntax1} \\
    &\Bunit,\ \Bmult,\ \scalarop,\ \Wcomult,\ \Wcounit \;\}  \label{eq:GLAsyntax2}
\end{align}
where $k\in\mathsf{k}$ is a \emph{scalar}.
Intuitively the black structure can be thought of as \emph{copying}, the white as \emph{adding}. This is borne out by the intended semantics, which we describe next.

The prop of \emph{linear relations} $\LinRel{\mathsf{k}}$ has as arrows $m\to n$ relations $R\subseteq \mathsf{k}^m\times\mathsf{k}^n$, which are $\mathsf{k}$-vector spaces; i.e.\ closed under $\mathsf{k}$-linear combinations. In other words, arrows $m\to n$ are linear subspaces of $\mathsf{k}^m\times\mathsf{k}^n$ considered as a $\mathsf{k}$-vector space. Composition is standard relational composition, $R\mathrel{;}S = \{\, (u,w) \;|\; \exists v.\; (u,v)\in R \ \wedge\ (v,w)\in S\,\}$.
It is simple to show linear relations are closed under composition.

The semantics of GLA is a prop morphism $[-]_{\mathsf{k}}\mathrel{:}\mathsf{GLA}_\mathsf{k} \rightarrow \LinRel{\mathsf{k}}$.
Since $\mathsf{GLA}_\mathsf{k}$ is free, it is enough to describe its action on the generators.
We do so below for the generators in~\ref{eq:GLAsyntax1};
the corresponding generators in \ref{eq:GLAsyntax2} are sent to the opposite relations. In each case the variables range over $\mathsf{k}$.
\[
\begin{array}{lllll}
    [\Bcomult]_\mathsf{k} = \left\{\,\left(x,\, \scriptstyle {\scriptstyle x \choose \scriptstyle x}\right)
    \,\right\},
    &
    [\Bcounit]_\mathsf{k} = \{\,(x,\, \bullet)
    \,\},
    &
    [\scalar]_\mathsf{k} = \{\,(x,\, kx)
    \,\},
    \\[5pt]
    {[}\Wmult]_\mathsf{k} = \left\{\,\left(\scriptstyle {\scriptstyle x \choose \scriptstyle y},\, x+y\right)
    \,\right\},
    &
    [\Wunit]_\mathsf{k} = \{\,(\bullet,\, 0)\}
\end{array}
\]
The associated theory $\IH$ characterises linear relations. We give a brief overview below:
\begin{itemize}
\item both monoids $(\Bmult,\,\Bunit)$ $(\Wmult,\,\Wunit)$ and comonoids $(\Bcomult,\,\Bcounit)$, $(\Wcomult,\,\Wunit)$ are commutative;
\item monoids and comonoids of the opposite colour satisfy the equations of commutative bialgebras;
\item monoids and comonoids of the same colour satisfy the extra special Frobenius equations;
\item to pass between black and white cups and caps is to compose with $-1$. We shall often draw the $-1$ scalar as $\antipode$.
\item All non-zero scalars are invertible, with the inverse
of $\scalar$ given by  $\scalarop$ for $k\neq 0$.
\end{itemize}
%

Graphical Affine Algebra was introduced in~\cite{bonchiGraphicalAffineAlgebra2019},
extending the expressivity of GLA to affine relations.
Let us recall the main concepts.
A \emph{translation} $v+V$ of a linear subspace $V$ by a vector $v$ is the set $v+V=\{\, v+w \;|\; w\in V\,\}$.
An \emph{affine subspace} $W$ 
is either empty, or it is the translation $v+V$ for some vector $v$ and subspace $V$. Note that the empty set is an affine subspace, but it is \emph{not} a linear subspace.

The prop of affine relations $\AffRel{\mathsf{k}}$ has as
affine subspaces of $\mathsf{k}^m\times\mathsf{k}^n$ as arrows $m\to n$. It is the case that the composition of two affine relations yields an affine relation, and can be shown using the notion of homogenisation~\cite[Proposition~6]{bonchiGraphicalAffineAlgebra2019}, but we will not delve into the details here.

\smallskip
On the syntactic side, we extend the signature~\eqref{eq:GLAsyntax1}, \eqref{eq:GLAsyntax2} with a single generator
\begin{equation}\label{eq:GAAsyntax}
\One
\end{equation}
obtaining a richer string diagrammatic syntax as the arrows of the free prop $\mathsf{GAA}$. Abusing notation, the semantics
$[-]_\mathsf{k} \mathrel{:} \mathsf{GAA} \rightarrow \AffRel{\mathsf{k}}$ is defined the same way as for $\mathsf{GLA}$ on the shared generators, and $[\One]_\mathsf{k} = \{\,(\bullet,\, 1)\}$. To characterise $\AffRel{\mathsf{k}}$ we add the following equations:
\[

\tikzset{x=1em, y=2.1ex}
\InputIfFileExists{one-copy.tikz}{}{\input{./tikz/one-copy.tikz}}
\tikzset{x=1em, y=1.5ex}
 \quad\myeq{dup}\quad 
\tikzset{x=1em, y=2.1ex}
\InputIfFileExists{one2.tikz}{}{\input{./tikz/one2.tikz}}
\tikzset{x=1em, y=1.5ex}
 \qquad

\tikzset{x=1em, y=2.1ex}
\InputIfFileExists{one-delete.tikz}{}{\input{./tikz/one-delete.tikz}}
\tikzset{x=1em, y=1.5ex}
 \quad\myeq{del}\quad 
\tikzset{x=1em, y=2.1ex}
\InputIfFileExists{empty-diag.tikz}{}{\input{./tikz/empty-diag.tikz}}
\tikzset{x=1em, y=1.5ex}
 \qquad

\tikzset{x=1em, y=2.1ex}
\InputIfFileExists{one-false.tikz}{}{\input{./tikz/one-false.tikz}}
\tikzset{x=1em, y=1.5ex}
 \quad\myeq{empty}\quad 
\tikzset{x=1em, y=2.1ex}
\InputIfFileExists{one-false-disconnect.tikz}{}{\input{./tikz/one-false-disconnect.tikz}}
\tikzset{x=1em, y=1.5ex}

\]

\medskip
Instead of relying purely on equational reasoning, it is often convenient to work with inequations, which can lead to shorter calculations, and allow us to identify some interesting, higher-level categorical structures. Instead of working with ordinary props, one works with \emph{ordered} props, that is props enriched over the category of posets. Similarly any equational theory can be presented as an \emph{inequational} theory in the obvious way, by replacing an equation by two inequations.

Indeed, $\LinRel{\mathsf{k}}$
and $\AffRel{\mathsf{k}}$ can be considered as ordered props by using set-theoretical inclusion of relations as the homset order. On the syntactic side, it suffices~\cite{bonchiRefinementSignalFlow2017} to add a single inequation:
\[
\Wunit \leq \Bunit
\]
The ordered setting also lends itself to higher level reasoning schema. In particular, $\LinRel{\mathsf{k}}$ is an abelian bicategory of relations and $\AffRel{\mathsf{k}}$ is a cartesian bicategory of relations, concepts developed in~\cite{carboniCartesianBicategories1987}.

\subsection{The prop of electrical circuits and its semantics}\label{prop-of-circuits}

We recall the string diagrammatic development of electrical circuits
from~\cite{bonchiGraphicalAffineAlgebra2019},
with minor modifications to suit our development in subsequent sections.
The prop $\circprop$ is free on the following signature:
\begin{equation}\label{eq:circpropsig}
\left\{\resistor ,\, \vsource ,\, \csource,\, \inductor,\, \capacitor \right\}_{R,L,C\in\R_+,\,V,I\in\R} \cup \quad \left\{\junction,\, \circcounit,\, \cojunction,\, \circunit \right\}
\end{equation}
where the parameters range over the reals. Arrows $m\to n$ of $\circprop$ represent open linear electrical circuits with $m$ open terminals on the left
and $n$ open terminals on the right. Generator $\resistor$ represents a \emph{resistor}, $\vsource$ a \emph{voltage source}, $\csource$ a \emph{current source},
$\inductor$ an \emph{inductor} and
$\capacitor$ a \emph{capacitor}.

Circuits in $\circprop$ are translated to $\mathsf{GAA}$ over $\R(x)$: the field of fractions of polynomials with real coefficients (see~\cite{coyaCircuitsBondGraphs2018,bonchiGraphicalAffineAlgebra2019}).
The semantics is a strict monoidal functor
$\mathcal{I}:\circprop \to \mathsf{GAA}$ where $\mathcal{I}(1)=2$ on objects: the idea is that every electrical wire is represented by \emph{two} $\mathsf{GAA}$ wires, the voltage wire on top and the current wire on the bottom.
We give the semantics in Figure~\ref{fig:circsemantics}, by showing the action on generators. Given circuits $c,d$ of $\circprop$ we write $c \semEq d$ when $\semElec{c}=\semElec{d}$ (in the equational theory of $\mathsf{GAA}$, or equivalently as affine relations in $\AffRel{\R(x)}$). Similarly, we write $c \semLeq d$ when $\semElec{c}\leq\semElec{d}$.

\begin{figure}\label{fig:semelec}
\newcommand{\electricalSemantics}[4]{
    \!\!
    \semElec{#1} = {#2}
    \ \mapsto \
    \left\{\scriptstyle \left({\scriptstyle #3}\right)\,\big|\, {\scriptstyle #4} \right\}
}
\newcommand{\elementSemantics}[3]{
    \electricalSemantics{#1}{#2}{{\phi_1 \choose i},{\phi_2 \choose i}}{#3}
}

\[
\elementSemantics
    {\resistor}
    {\lower10pt\hbox{$\includegraphics[height=1cm]{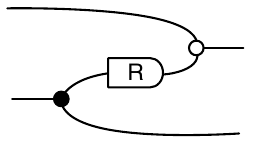}$}}
    {\phi_2-\phi_1 = Ri}
\]
\[
\elementSemantics
    {\vsource}
    {\lower10pt\hbox{$\includegraphics[height=.8cm]{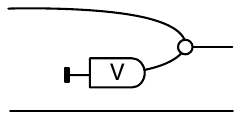}$}}
    {\phi_2-\phi_1=V}
,\
\elementSemantics
    {\csource}
    {\lower10pt\hbox{$\includegraphics[height=.8cm]{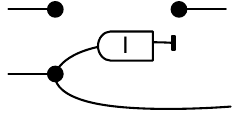}$}}
    {i=I}
\]
\[
\!\!\!\!\!\!\!
\elementSemantics
    {\inductor}
    {\!\!\!\lower10pt\hbox{$\includegraphics[height=.8cm]{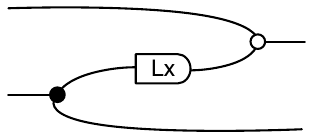}$}}
    {\phi_2-\phi_1 = Lxi}
,\
\elementSemantics
    {\capacitor}
    {\!\!\!\lower10pt\hbox{$\includegraphics[height=.8cm]{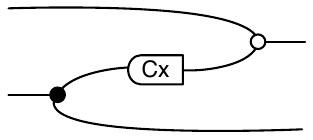}$}}
    {i = Cx(\phi_2-\phi_1)}
\]
\[
\electricalSemantics
    {\junction}
    {\lower10pt\hbox{$\includegraphics[height=.8cm]{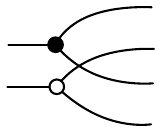}$}}
    {{\phi \choose i_1}, \scriptscriptstyle \left(\begin{smallmatrix} \scriptscriptstyle \phi \\ \scriptscriptstyle i_2 \\ \scriptscriptstyle \phi \\ \scriptscriptstyle i_3 \end{smallmatrix}\right)}
    {i_1 + i_2 + i_3 = 0}
,\
\electricalSemantics
    {\cojunction}
    {\lower10pt\hbox{$\includegraphics[height=.8cm]{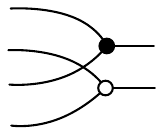}$}}
    {{\scriptscriptstyle \left(\begin{smallmatrix} \scriptscriptstyle \phi \\ \scriptscriptstyle i_2 \\ \scriptscriptstyle \phi \\ \scriptscriptstyle i_3 \end{smallmatrix}\right)}, {\phi \choose i_1}}
    {i_1 + i_2 + i_3 = 0}
\]
\[
\electricalSemantics
    {\circcounit}
    {\lower5pt\hbox{$\includegraphics[height=.5cm]{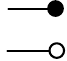}$}}
    {{\phi \choose i}, \bullet}
    {i = 0}
,\quad\
\electricalSemantics
    {\circunit}
    {\lower5pt\hbox{$\includegraphics[height=.5cm]{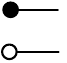}$}}
    {\bullet, {\phi \choose i}}
    {i = 0}
\]
\caption{Compositional semantics of circuits.\label{fig:circsemantics}}
\end{figure}

\section{The Impedance Calculus}\label{sec:imp-calculus}

We now exploit a pattern of the semantics in Figure~\ref{fig:circsemantics} 
 to simplify the passage between circuits and $\mathsf{GAA}$. This results in the \emph{impedance calculus}, which can be used to simplify diagrammatic reasoning on circuits.

We extend the syntax~\eqref{eq:circpropsig} of $\circprop$ with \emph{impedance boxes}---illustrated in~\eqref{eq:impedanceBox}---parametrised with respect to arbitrary $\mathsf{GAA}$ circuits of type $\sort{1}{1}$: that is, with one wire on the left and one on the right. We then extend the semantic mapping $\semElec{\cdot}$ to cover impedance boxes, as below right.
\begin{equation}\label{eq:impedanceBox}
    \impedanceBox,
    \qquad \qquad
    \semElec{\impedanceBox} \quad = \quad \lower10pt\hbox{$\includegraphics[height=1cm]{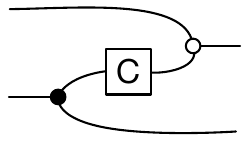}$}
\end{equation}

The following are now easy derivations in the equational theory of GAA:
\[
\begin{array}{rccrccrcc}
    \wire       &\semEq& \lower13pt\hbox{$\includegraphics[height=1cm]{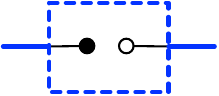}$},
    &\vsource   &\semEq& \lower13pt\hbox{$\includegraphics[height=1cm]{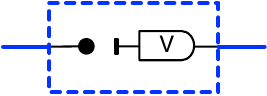}$},
    &\csource   &\semEq& \lower13pt\hbox{$\includegraphics[height=1cm]{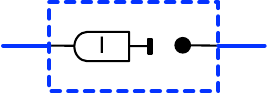}$},
    \\
    \resistor   &\semEq& \lower13pt\hbox{$\includegraphics[height=1cm]{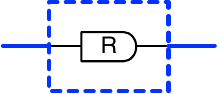}$},
    &\inductor  &\semEq& \lower13pt\hbox{$\includegraphics[height=1cm]{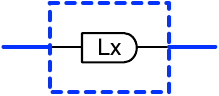}$},
    &\capacitor &\semEq& \lower13pt\hbox{$\includegraphics[height=1cm]{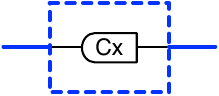}$}.
\end{array}
\]

We now prove results that give the impedance calculus its power and allow us to manipulate impedance
boxes within circuits.
Henceforward we will use the syntactic sugar
$\lower6pt\hbox{$\includegraphics[height=.8cm]{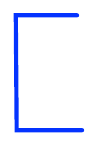}$} \Defeq \lower6pt\hbox{$\includegraphics[height=.8cm]{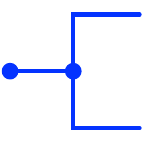}$}$
and
$\lower6pt\hbox{$\includegraphics[height=.8cm]{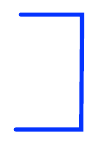}$} \Defeq \lower6pt\hbox{$\includegraphics[height=.8cm]{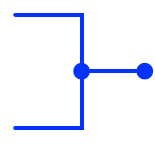}$}$.

\begin{lemma}\label{lem:impedanceProperties}~
\[\begin{array}{rcccrccc}
    \text{(i)}&
    \lower8pt\hbox{$\includegraphics[height=.6cm]{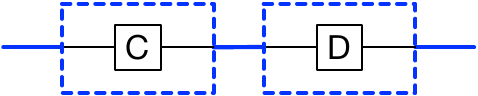}$}
    &\semEq&
    \lower12pt\hbox{$\includegraphics[height=1cm]{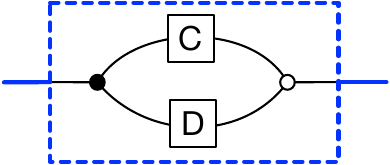}$}
    &
    \text{(ii)}&
    \lower15pt\hbox{$\includegraphics[height=1.3cm]{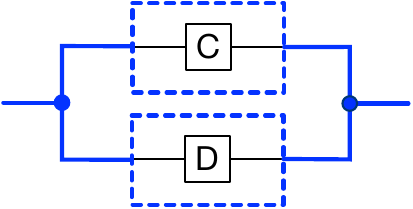}$}
    &\semEq &
    \lower12pt\hbox{$\includegraphics[height=1cm]{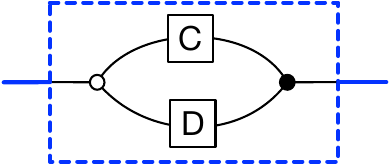}$}
    \\
    \text{(iii)}&
    \lower17pt\hbox{$\includegraphics[height=1.2cm]{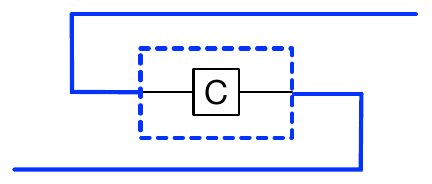}$}
    &\semEq&
    \lower8pt\hbox{$\includegraphics[height=.6cm]{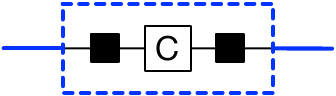}$}
    &
    \text{(iv)}&
    \semElec{\lower8pt\hbox{$\includegraphics[height=.6cm]{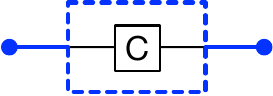}$}}
    &=&
    \lower4pt\hbox{$\includegraphics[height=.4cm]{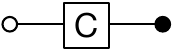}$}
\end{array}\]
\end{lemma}

\begin{proof}
    \[
    \text{(i)}: \lower18pt\hbox{$\includegraphics[height=1.4cm]{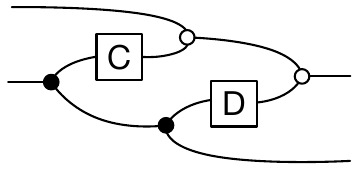}$}
    \quad = \quad  \lower20pt\hbox{$\includegraphics[height=1.6cm]{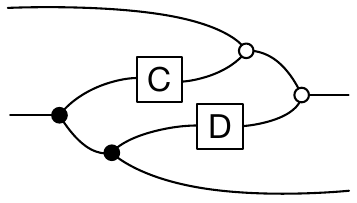}$}
    \quad = \quad  \lower18pt\hbox{$\includegraphics[height=1.4cm]{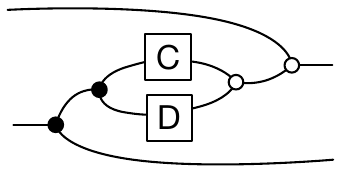}$}
    \]
    \[
    \text{(ii)}: \lower22pt\hbox{$\includegraphics[height=1.8cm]{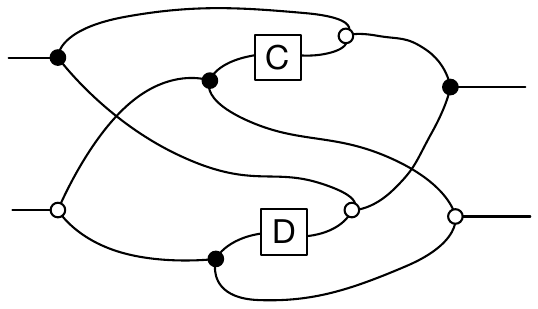}$}
    \quad = \quad \lower22pt\hbox{$\includegraphics[height=1.8cm]{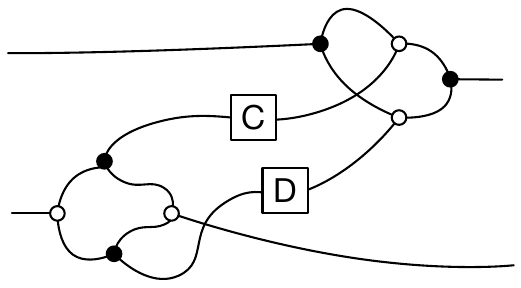}$}
    \quad = \quad \lower22pt\hbox{$\includegraphics[height=1.8cm]{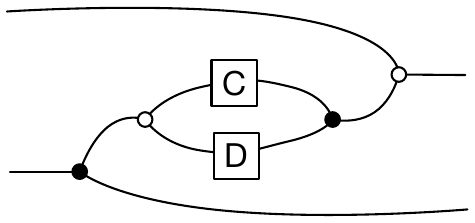}$}
    \]
    \begin{multline*}
    \text{(iii)}: \lower22pt\hbox{$\includegraphics[height=1.8cm]{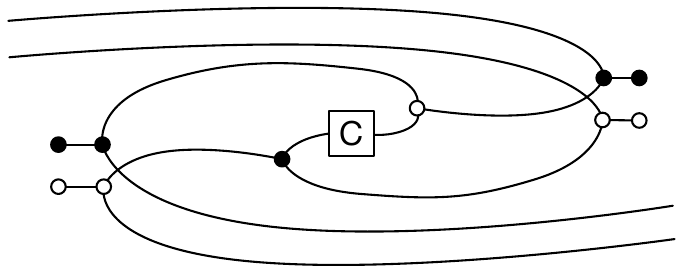}$}
    \quad = \quad  \lower22pt\hbox{$\includegraphics[height=1.8cm]{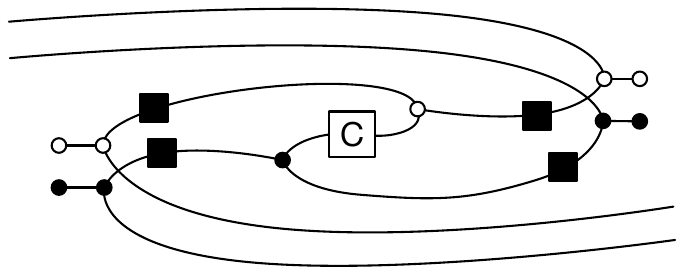}$} \\
     = \quad  \lower22pt\hbox{$\includegraphics[height=1.8cm]{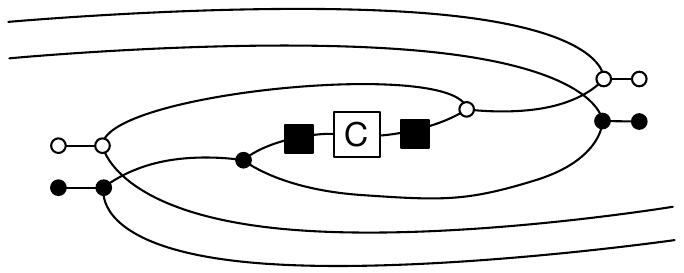}$}
    \quad = \quad \lower13pt\hbox{$\includegraphics[height=1.2cm]{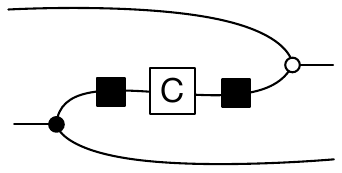}$}
    \end{multline*}
    \[
    \text{(iv)}: \lower13pt\hbox{$\includegraphics[height=1.2cm]{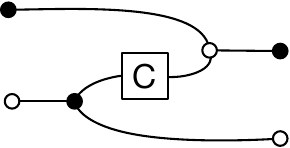}$} \quad=\quad \lower4pt\hbox{$\includegraphics[height=.4cm]{graffles/closedRHS.pdf}$}
    \]
\end{proof}

Using \Cref{lem:impedanceProperties} \text{(iii)} we can immediately derive several useful properties of circuits:
\begin{corollary}~
    \begin{enumerate}[(i)]
    \item Resistors, inductors and capacitors are ``directionless'':
    \[
    \lower15pt\hbox{$\includegraphics[height=1.3cm]{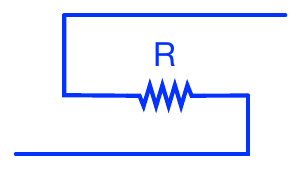}$}
    \semEq
    \resistor,
    \qquad
    \lower15pt\hbox{$\includegraphics[height=1.3cm]{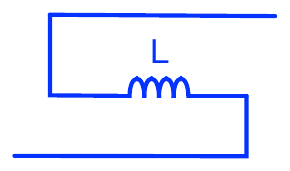}$}
    \semEq
    \inductor,
    \qquad
    \lower15pt\hbox{$\includegraphics[height=1.3cm]{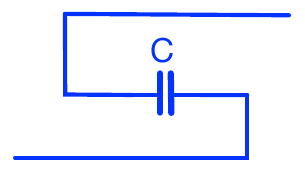}$}
    \semEq
    \capacitor.
    \]
    \item Reversing the direction of voltage and current sources flips polarities:
    \[
    \lower15pt\hbox{$\includegraphics[height=1.3cm]{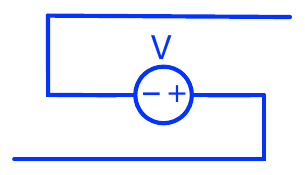}$}
    \semEq\
    \lower5pt\hbox{$\includegraphics[height=.8cm]{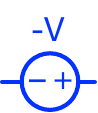}$}
    \ \Defeq\
    \lower5pt\hbox{$\includegraphics[height=.8cm]{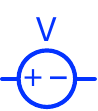}$},
    \qquad\qquad
    \lower15pt\hbox{$\includegraphics[height=1.3cm]{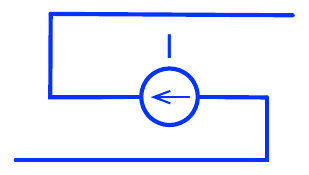}$}
    \semEq\
    \lower5pt\hbox{$\includegraphics[height=.8cm]{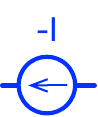}$}
    \ \Defeq\
    \lower5pt\hbox{$\includegraphics[height=.8cm]{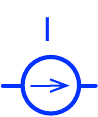}$}.
    \]
    \end{enumerate}
\end{corollary}

The impedance calculus is useful for proving
circuit equivalences. The following proposition are just a few examples of classic equivalences
one would find in any textbook.

\begin{proposition}~
\begin{enumerate}[(i)]
\item $\lower5pt\hbox{$\includegraphics[height=.8cm]{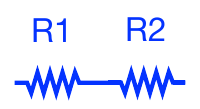}$} \semEq
       \lower5pt\hbox{$\includegraphics[height=.8cm]{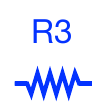}$}$ where
      $\lower3pt\hbox{$\includegraphics[height=.4cm]{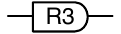}$} =
       \lower9pt\hbox{$\includegraphics[height=.8cm]{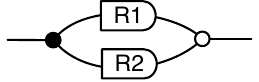}$} =
       \lower3pt\hbox{$\includegraphics[height=.4cm]{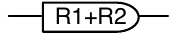}$}$
\item $\lower14pt\hbox{$\includegraphics[height=1.4cm]{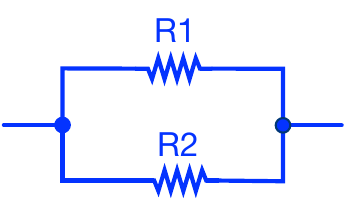}$} \semEq
       \lower5pt\hbox{$\includegraphics[height=.8cm]{graffles/r3.pdf}$}$ where
      $\lower3pt\hbox{$\includegraphics[height=.4cm]{graffles/r3gla.pdf}$} =
       \lower9pt\hbox{$\includegraphics[height=.8cm]{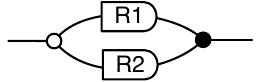}$} =
       \lower2.5pt\hbox{$\includegraphics[height=.3cm]{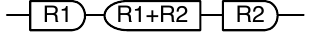}$}$
\item $\lower14pt\hbox{$\includegraphics[height=1.4cm]{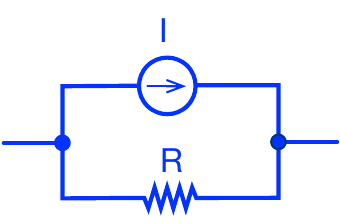}$} \semEq
       \lower5pt\hbox{$\includegraphics[height=.8cm]{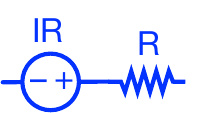}$}$
\item $\lower14pt\hbox{$\includegraphics[height=1.4cm]{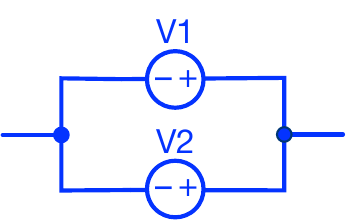}$} \semEq
       \lower5pt\hbox{$\includegraphics[height=.8cm]{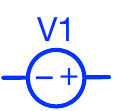}$}$\ \ if $V1=V2$,
        otherwise its semantics is $\emptyset$ (the empty relation)
\end{enumerate}
\end{proposition}

It is useful to contrast our treatment with the classical approach.
Parts
(i) and (iii) are standard and often-used equivalences. Part
(ii) is known classically,
but $R3$ is typically given a formula like $R1 R2/(R1+R2)$.
Classical formulas, however, do not work for all values of $R1$ and $R2$,
whereas the graphical one does.
Given that it mirrors the case
for resistors in series,
we argue that the graphical formula is the more natural one.
Finally, in (iv), the empty case is usually excluded by classical treatments.
A textbook deems such a circuit \textit{degenerate}
and ignores that case when proving theorems.
In $\mathsf{GAA}$ however, the empty relation is a first-class citizen
and our theorems uniformly include the empty case as well.

\begin{proof}
Part (i) is a simple exercise in the use of the impedance calculus:
\[
\lower5pt\hbox{$\includegraphics[height=.8cm]{graffles/r1r2series.pdf}$} \semEq
\lower12pt\hbox{$\includegraphics[height=1cm]{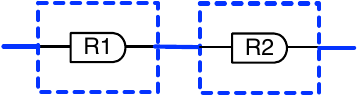}$} \semEq
\lower12pt\hbox{$\includegraphics[height=1cm]{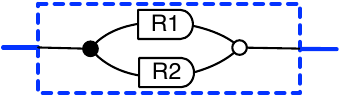}$} \semEq
\lower5pt\hbox{$\includegraphics[height=.8cm]{graffles/r3.pdf}$}.
\]

For part (ii), we have $\lower14pt\hbox{$\includegraphics[height=1.4cm]{graffles/r1r2par.pdf}$} \semEq
\lower20pt\hbox{$\includegraphics[height=1.7cm]{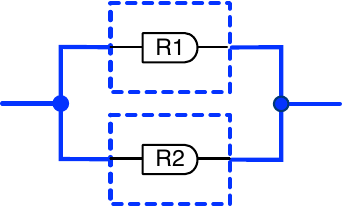}$} \semEq
\lower12pt\hbox{$\includegraphics[height=1cm]{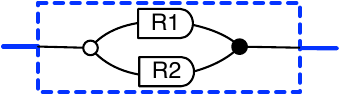}$}.$
Now
\begin{multline*}
\lower9pt\hbox{$\includegraphics[height=.8cm]{graffles/r1r2pargla.pdf}$}
=
\lower9pt\hbox{$\includegraphics[height=.8cm]{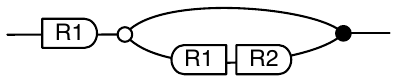}$}
=
\lower9pt\hbox{$\includegraphics[height=.8cm]{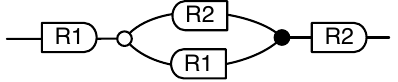}$}
=
\lower2.5pt\hbox{$\includegraphics[height=.3cm]{graffles/thisismyfinalform.pdf}$}
\end{multline*}
extracts the classical formula: because $R1$ and $R2$ are nonnegative,
either $R1+R2\neq 0$ and
$\lower2.5pt\hbox{$\includegraphics[height=.3cm]{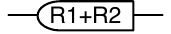}$}$
is a scalar, or $R1=R2=0$ and the formula is equal to $0$.
In both cases the result is a scalar.

Part (iii) is another simple calculation:
\begin{multline*}
\lower12pt\hbox{$\includegraphics[height=1.4cm]{graffles/csourceparresistor.pdf}$}
\semEq \lower17pt\hbox{$\includegraphics[height=1.4cm]{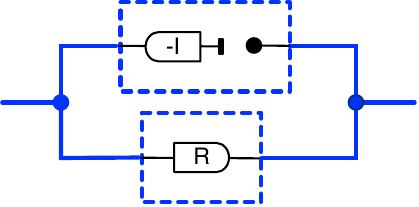}$}
\semEq \lower14pt\hbox{$\includegraphics[height=1.3cm]{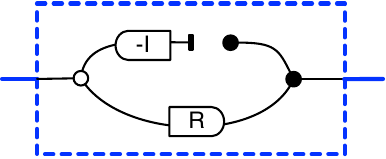}$}
\semEq \lower14pt\hbox{$\includegraphics[height=1.3cm]{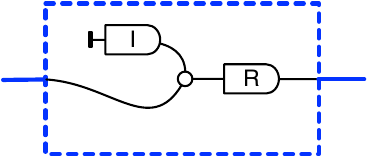}$} \\
\semEq \lower14pt\hbox{$\includegraphics[height=1.3cm]{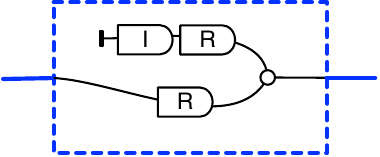}$}
\semEq \lower12pt\hbox{$\includegraphics[height=.9cm]{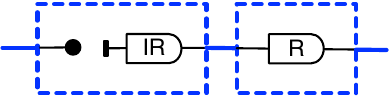}$}
\semEq \lower5pt\hbox{$\includegraphics[height=.8cm]{graffles/vsourceseriesresistor.pdf}$}
\end{multline*}

For part (iv), we can simplify using the impedance calculus as follows:
\[
\lower14pt\hbox{$\includegraphics[height=1.4cm]{graffles/vsourceparvsource.pdf}$}
\semEq \lower18pt\hbox{$\includegraphics[height=1.4cm]{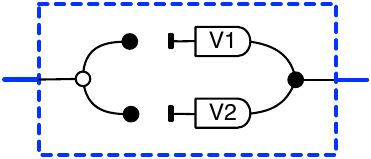}$}
\semEq \lower18pt\hbox{$\includegraphics[height=1.4cm]{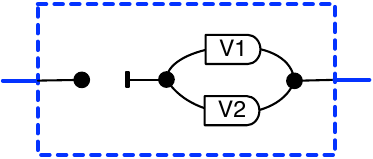}$}
\]
Now $\lower9pt\hbox{$\includegraphics[height=.9cm]{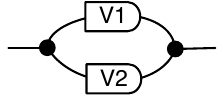}$}$
is just $V1$ if $V1=V2$, and $\lower1pt\hbox{$\includegraphics[height=.2cm]{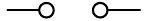}$}$
otherwise. In that case, the circuit evaluates to
$\lower14pt\hbox{$\includegraphics[height=1cm]{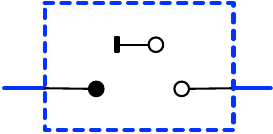}$}$ which denotes the empty relation.
\end{proof}
We shall see that, as a consequence of our Representation Theorem (Theorem~\ref{thm:representation}), \emph{any} $\sort{\bluebullet}{\bluebullet}$ circuit can be represented by an impedance box.

\section{Measuring Closed Circuits}

Thus far, we have kept the language of circuits $\circprop$ and the language of
$\mathsf{GAA}$ neatly separated by impedance boxes~\eqref{eq:impedanceBox} and the prop morphism
$\semElec{\cdot}\mathrel{:}\circprop \to \mathsf{GAA}$. It is time to tear down the wall.

While impedance boxes are useful for open circuit calculations, engineers
 often study closed circuits.
The problem is that these are mapped to $0 \to 0$ affine relations,
and there are only two: the singleton and the empty set.
In other words, the mapping only tells us whether the circuit has a solution,
and nothing else. Indeed, no interesting behavior can be observed in a closed circuit in our existing framework.

Classically, a closed circuit is annotated with names for the currents and voltage differences of interest.
In some presentations voltmeters and ammeters are even added as explicit elements.
This is exactly what we do: add meters to our syntax.
Instead of assigning names, however, our meters have an outgoing ``information'' wire, thought of as an ordinary $\mathsf{GAA}$ wire. Information wires
allow us to probe closed circuits, i.e.\ those with no open \textit{electric} wires,
 and to state and prove several theorems of circuit theory. Our diagrams will now have two types of wires: electric wires and information wires. Instead of keeping the props $\circprop$ and $\mathsf{GAA}$ separate, we unite the two and work with a coloured (multisorted) prop $\ecircprop$.

The coloured prop $\ecircprop$ has as objects words over the alphabet $\{\bluebullet,\,\bullet\}$. It is freely generated
by the union of all the generators \eqref{eq:GLAsyntax1}, \eqref{eq:GLAsyntax2}, \eqref{eq:GAAsyntax}
of $\mathsf{GAA}$, the generators \eqref{eq:circpropsig} of $\circprop$, and the following:
\[
\left\{\,\voltmeter,\,\ampmeter,\,\cvsource,\,\ccsource,\,\lower12pt\hbox{$\includegraphics[height=1.4cm]{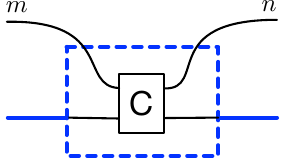}$}\,\right\}_{m,n\in \mathbb{N},\, c\from m+1\rightarrow n+1 \text{ in }\mathsf{GAA}}.
\]
These are a \emph{voltmeter} $\voltmeter\typ\sort{\bluebullet}{\bullet\bluebullet}$, \emph{ammeter} $\ampmeter\typ\sort{\bluebullet}{\bullet\bluebullet}$, \emph{controlled voltage source}
$\cvsource\typ\sort{\bullet\bluebullet}{\bluebullet}$, \emph{controlled current source} $\ccsource\typ\sort{\bullet\bluebullet}{\bluebullet}$ and a \emph{generalised impedance box}
$\sort{\bullet^m\bluebullet}{\bullet^n\bluebullet}$, parametrised over arbitrary $\mathsf{GAA}$ diagrams $c\mathrel{:}m+1\to n+1$.

We abuse notation and denote the translation of $\ecircprop$ to $\mathsf{GAA}$ by $\mathcal{I}\from\ecircprop\rightarrow\mathsf{GAA}$.
On objects,
$\semElec{\bluebullet} = 2$, since---as before---an electrical wire is represented by two wires in $\mathsf{GAA}$,
but the information wire ``is'' a GAA wire, i.e.\ $\semElec{\bullet} = 1$. Next, $\mathcal{I}$ acts on the generators of $\mathsf{GAA}$ as identity and on the basic electrical components~\eqref{eq:circpropsig} as described in Figure~\ref{fig:circsemantics}. We delay the translations of $\{\,\voltmeter,\,\ampmeter,\,\cvsource,\,\ccsource\,\}$ to later in this section and first focus on the generalised impedance boxes. Their translation is:
\[
\semElec{\lower12pt\hbox{$\includegraphics[height=1.4cm]{graffles/mnimpedanceBox.pdf}$}}
\quad = \quad
\lower20pt\hbox{$\includegraphics[height=1.8cm]{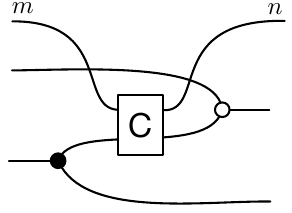}$}.
\]


We again use $\semEq$ and $\semLeq$ as circuit relations to mean $=$ and $\leq$ on their translations to $\mathsf{GAA}$, using its (in)equational theory.
The properties of Lemma~\ref{lem:impedanceProperties} easily generalise to these extended impedance boxes. We omit the details, and only mention one simple, but useful fact: $\mathsf{GAA}$ diagrams can breach 
impedance boxes. Suppose that $d\from m\rightarrow m'$, $e\from n\rightarrow n'$ in $\ecircprop$ are built up using only the $\mathsf{GAA}$ generators, then:
\begin{observation}
\[
\lower20pt\hbox{$\includegraphics[height=1.5cm]{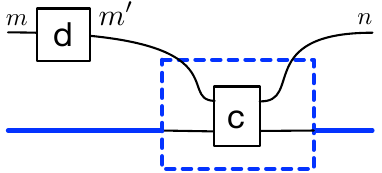}$}
\semEq
\lower20pt\hbox{$\includegraphics[height=1.7cm]{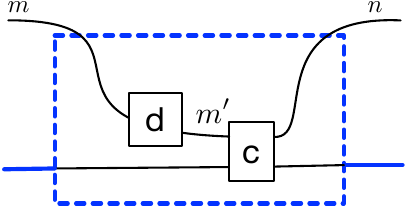}$}
\ \text{and}\
\lower20pt\hbox{$\includegraphics[height=1.7cm]{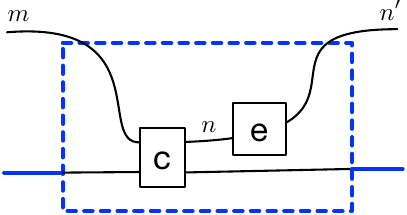}$}
\semEq
\lower20pt\hbox{$\includegraphics[height=1.5cm]{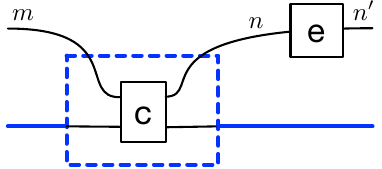}$}.
\]
\end{observation}

\subsection{Meters and controlled sources}

We now turn to the translations of the voltmeter and the ammeter.
\[
\semElec{\voltmeter} \quad = \quad \lower15pt\hbox{\includegraphics[height=1.5cm]{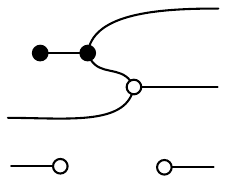}}
\qquad\qquad
\semElec{\ampmeter} \quad = \quad \lower15pt\hbox{\includegraphics[height=1.5cm]{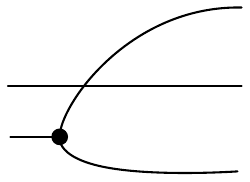}}
\]
Note that
orientation
follows the conventions
used so far:
current and voltage difference are right to left.

We can finally observe the behavior of closed circuits!
Consider a simple circuit, with just a battery and a resistor. Having no
open electrical wires, it is closed, but has one outgoing information wire.
Mapping to $\mathsf{GAA}$, we get a diagram with one output, giving us the computed value for the current.

\begin{multline*}
\semElec{\lower15pt\hbox{\includegraphics[height=1.5cm]{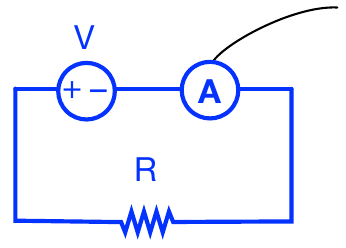}}}
=\
\lower25pt\hbox{\includegraphics[height=2cm]{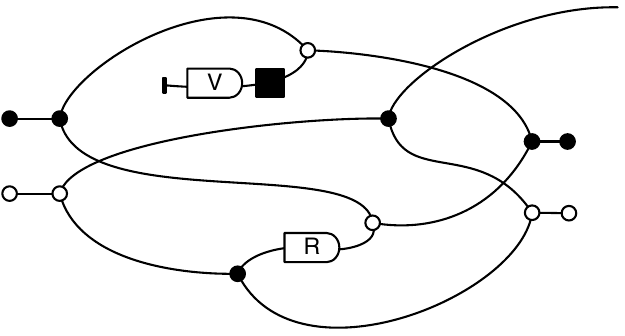}}
\!\!=
\!\lower22pt\hbox{\includegraphics[height=1.8cm]{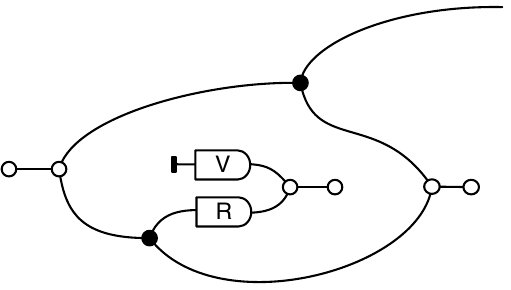}}
\!\!\!\!\!=
\lower12pt\hbox{\includegraphics[height=1.2cm]{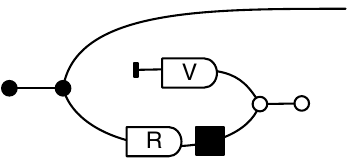}}
\!\!\!\!=\
\lower2pt\hbox{\includegraphics[height=.25cm]{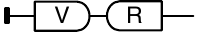}}
\end{multline*}
Although elementary enough above, these kinds of calculations can be drastically simplified using the impedance calculus by keeping the GAA diagrams small, as we shall demonstrate in the following.

\medskip
First let us turn to controlled sources.
Those are like independent sources, but their value is controlled by an information wire.
Their translations to $\mathsf{GAA}$ are given below.
\[
\semElec{\cvsource} \quad = \quad \lower12pt\hbox{\includegraphics[height=1.2cm]{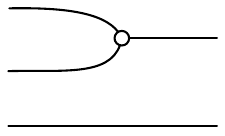}}\ ,\,\
\semElec{\ccsource} \quad = \quad \lower12pt\hbox{\includegraphics[height=1.2cm]{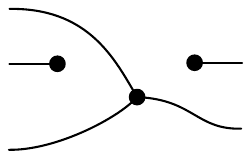}}.
\]

As in the previous section, it is usually simpler to combine impedances directly
rather than map them fully to GAA. This can be extended to work with meters too, giving
a convenient way to solve many circuits.
Using the generalised impedance boxes, we have:
\[
\voltmeter \semEq \lower10pt\hbox{\includegraphics[height=1.1cm]{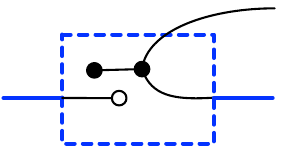}}
\ \quad\
\ampmeter \semEq \lower10pt\hbox{\includegraphics[height=1.1cm]{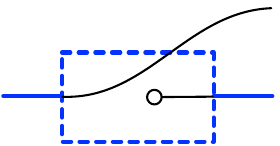}}
\ \quad\
\ccsource \semEq \lower10pt\hbox{\includegraphics[height=1.1cm]{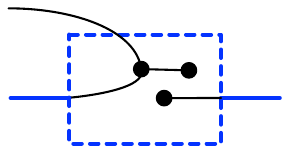}}
\ \quad\
\cvsource \semEq \lower10pt\hbox{\includegraphics[height=1.1cm]{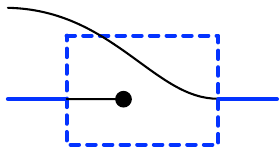}}
\]

We can redo our simple example using this new technology:

\[
\lower10pt\hbox{\includegraphics[height=1.3cm]{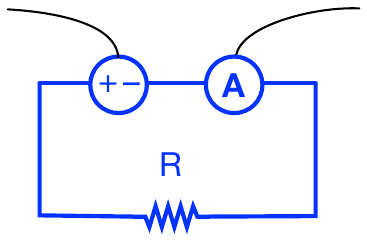}}
\semEq
\lower16pt\hbox{\includegraphics[height=1.6cm]{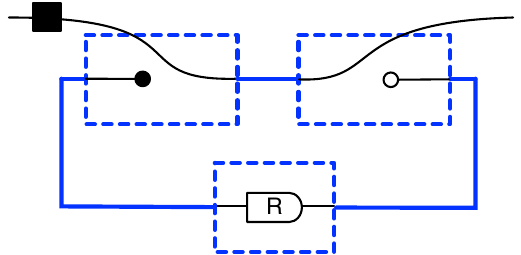}}
\semEq
\lower18pt\hbox{\includegraphics[height=1.8cm]{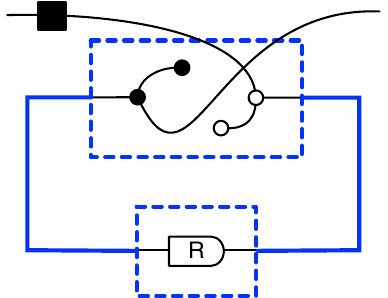}}
\semEq
\lower17pt\hbox{\includegraphics[height=1.7cm]{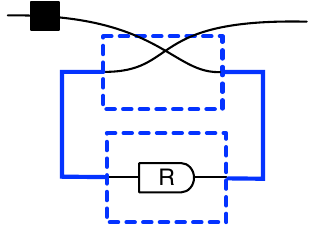}}
\semEq
\lower12pt\hbox{\includegraphics[height=1.2cm]{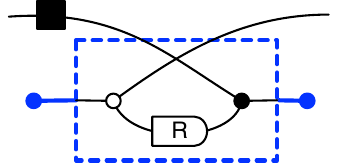}}
\]
and so
\[
\semElec{\lower14pt\hbox{\includegraphics[height=1.3cm]{graffles/simplified1.pdf}}}
=
\lower9pt\hbox{\includegraphics[height=.9cm]{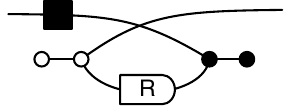}}
=
\lower3pt\hbox{\includegraphics[height=.3cm]{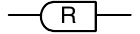}}
\]

An important property of meters and sources is that they can be considered one at a time.
This is often used classically to keep calculations manageable.
The following two theorems make this 
precise.
\subsection{Independent Measurement Theorem}

The first theorem is quite intuitive: 
measuring somewhere in a circuit ought not to affect measurements elsewhere.
In other words, we can extract the full behavior of a circuit by considering measurements one at a time.
In fact, this is so natural that it is just assumed to be true classically,
and not mentioned in textbooks.
We are being more rigorous here: 
the result is more subtle than one might think.

To get a single measurement we discard (i.e.\ plug $\Bcounit$ into) all meters but one.
It is easy to show that an ignored ammeter is equivalent to a wire, and an ignored voltmeter to an open circuit.
\[
\lower5pt\hbox{\includegraphics[height=.7cm]{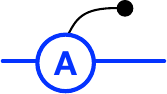}}
 \ \semEq\ \wire \qquad
\lower5pt\hbox{\includegraphics[height=.7cm]{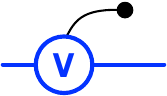}}
 \ \semEq\ \opencircuit
\]

\begin{theorem}[Independent Measurement Theorem]\label{thm:im}
    Given a closed circuit $C$ with $n$ meter outputs,
    the behavior of the $n$ measurements simultaneously is a subset of the behavior
    of one measurement at a time, ignoring the others.
    The proof is a simple derivation, using cartesian bicategory structure:
    \[
    \lower25pt\hbox{\includegraphics[height=2.1cm]{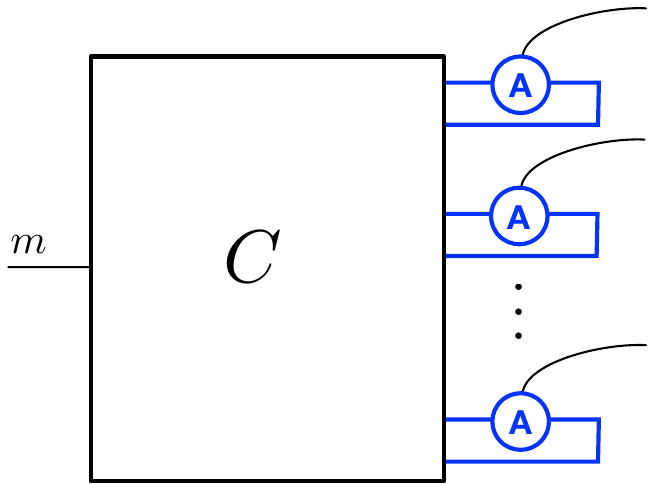}}
    \ = \
    \lower13pt\hbox{\includegraphics[height=1cm]{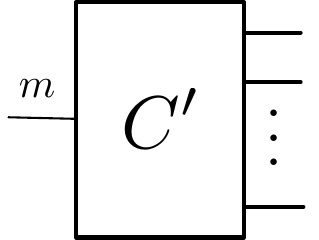}}
    \semEq
    \lower35pt\hbox{\includegraphics[height=2.8cm]{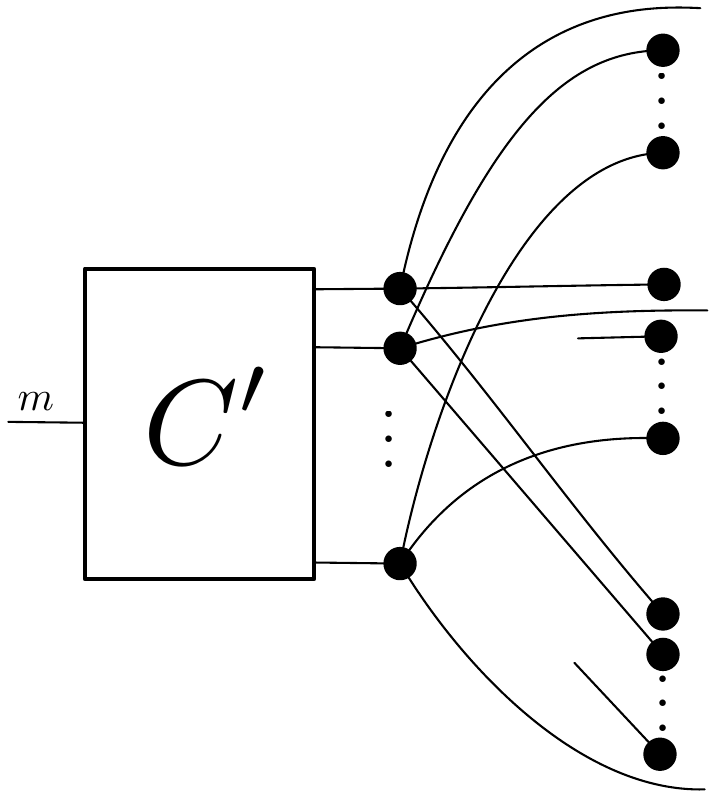}}
    \ \semLeq \
    \lower33pt\hbox{\includegraphics[height=2.5cm]{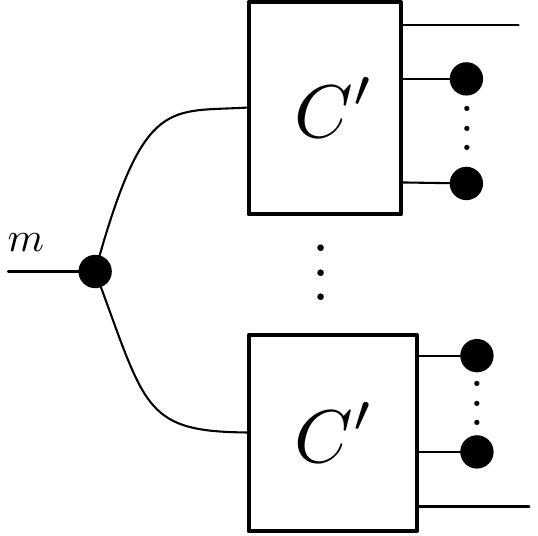}}
    \ = \
    \lower55pt\hbox{\includegraphics[height=4.5cm]{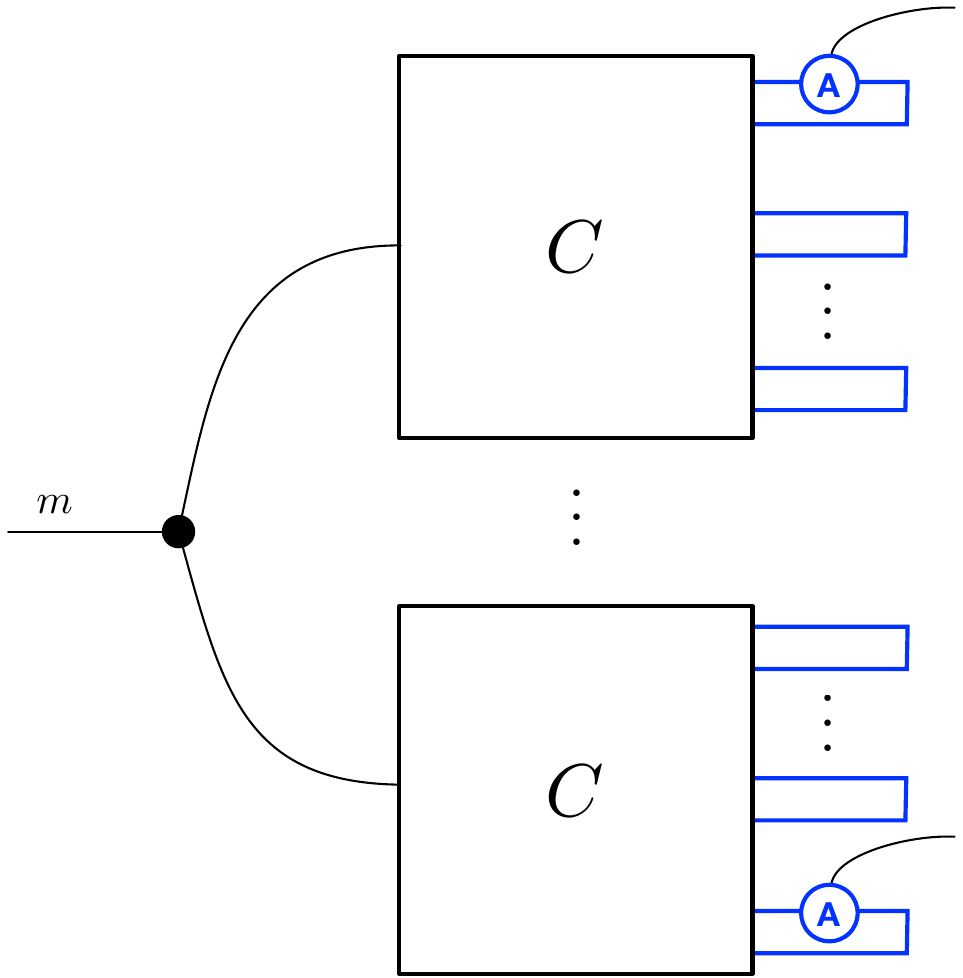}}
    \]
    The diagram above illustrates this for ammeters,
    but it applies to any mixture of meters,
    remembering that a discarded voltmeter is an open circuit.
    If moreover each individual measurement is a function of the inputs
    (i.e.\ single-valued and total), then the inclusion is an equality.
\end{theorem}

The inclusion above is an equality for all ``well-behaved'' circuits.
Classically, equality is postulated. 
It is, of course, reasonable to do so:
if the inclusion is strict,
it must be the case that
either \textit{(i)} that some meter can return a nonzero value with all sources off,
i.e.\ it is not measuring anything physical,
or \textit{(ii)} that some settings of the sources are disallowed and can cause the semantics to be empty,
which is degenerate.
For example, in the following example we have accidentally short-circuited the source,
 meaning that we cannot turn it on with a nonzero value.
This, indeed, is a degenerate case and we do not have equality.
\[
\lower25pt\hbox{\includegraphics[height=2cm]{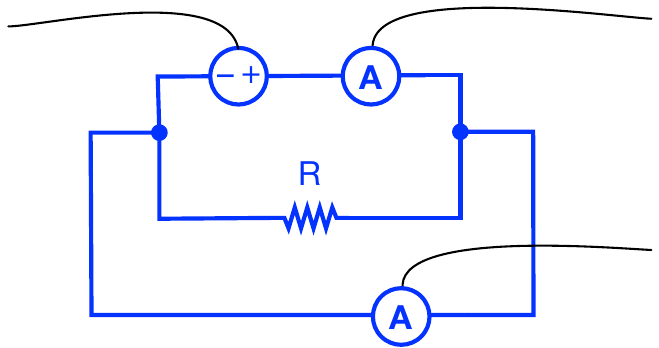}}
\semEq\
\lower10pt\hbox{\includegraphics[height=.8cm]{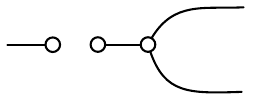}}
\ \neq\
\lower10pt\hbox{\includegraphics[height=.8cm]{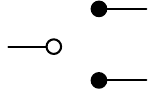}}
\ \semEq
\lower50pt\hbox{\includegraphics[height=3.5cm]{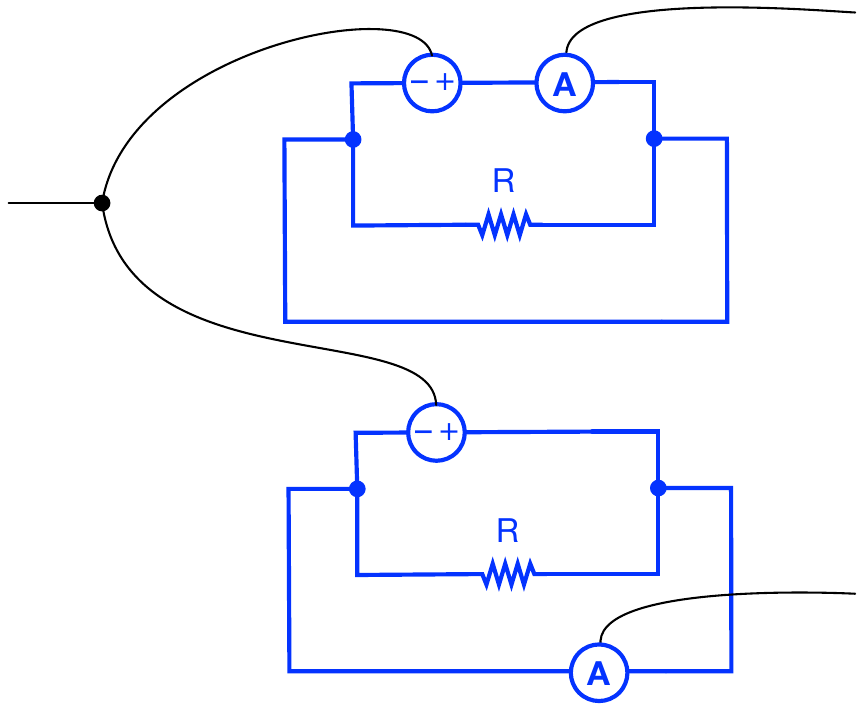}}
\]
It does not appear easy to characterise graphically
the circuits with such behavior.
Instead we give the result in its general form,
trusting that equality is easy to notice when calculating the right-hand side.

\subsection{Superposition Theorem}

The second theorem is amongst the most useful results of classical circuit theory.
The superposition theorem ensures that the behavior of a circuit with multiple sources
can be calculated from its behaviors with one source turned on at a time.

To do this, we set all sources but one to zero by plugging $\Wunit$.
It is easy to show that a zero current source is equivalent to a disconnected circuit,
and a zero voltage source is equivalent to a plain wire.

\[
\lower5pt\hbox{\includegraphics[height=.7cm]{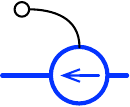}}
 \ \semEq\ \opencircuit \qquad
\lower5pt\hbox{\includegraphics[height=.7cm]{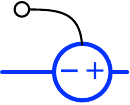}}
 \ \semEq\ \wire
\]

In a precise sense it is the dual of the independent measurement theorem:
they are related by swapping colors and vertical reflection,
which is a powerful operation of $\mathsf{GLA}$
that often generates elegant dualities.

\begin{theorem}[Superposition Theorem]\label{thm:superposition}
    The behaviour of closed circuit $C$ with $m$ source inputs and no independent sources
     is a superset of the sum of its behaviors
     with one source turned on at a time:
    \[
    \lower55pt\hbox{\includegraphics[height=4.5cm]{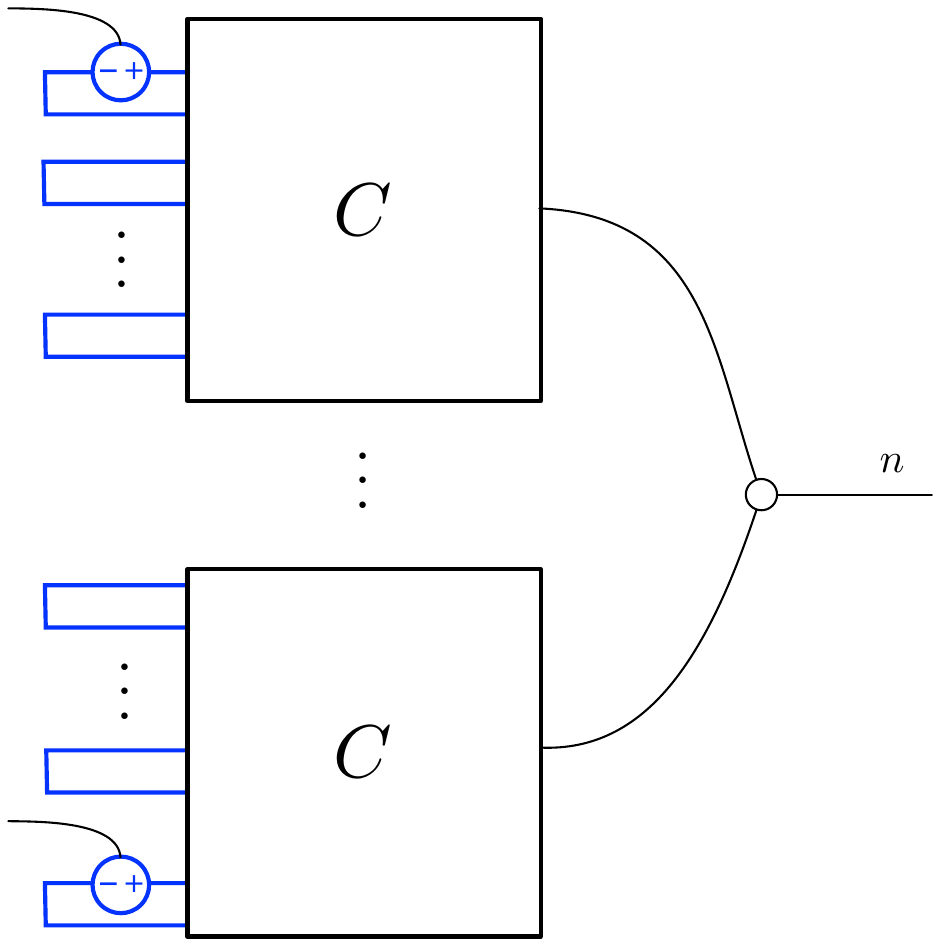}}
    \ = \
    \lower33pt\hbox{\includegraphics[height=2.5cm]{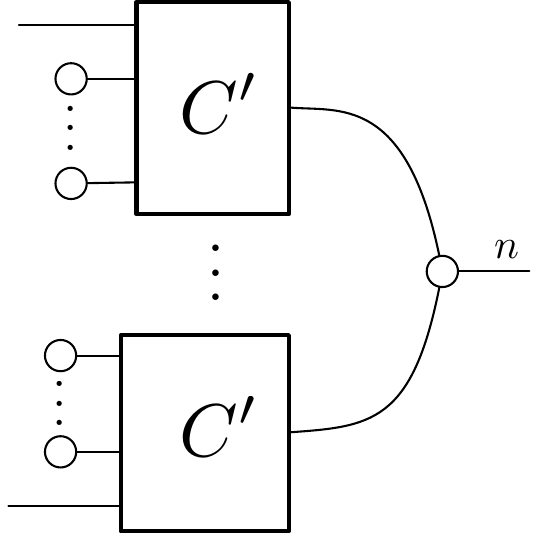}}
    \ \semLeq \
    \lower35pt\hbox{\includegraphics[height=2.8cm]{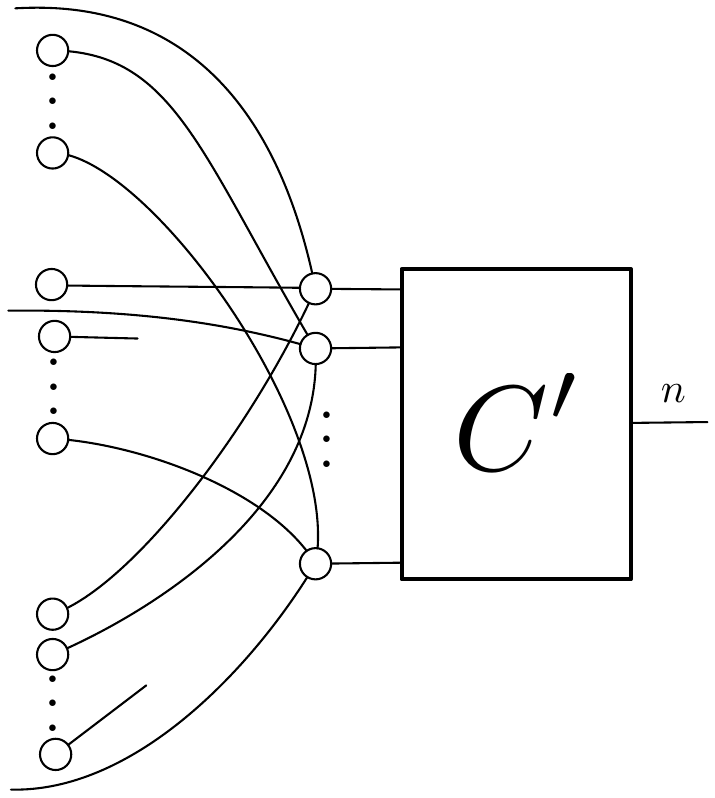}}
    \semEq
    \lower13pt\hbox{\includegraphics[height=1cm]{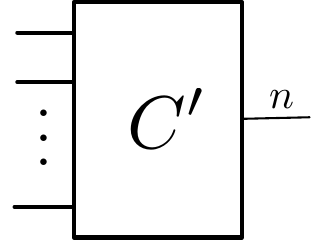}}
    \ = \
    \lower25pt\hbox{\includegraphics[height=2.1cm]{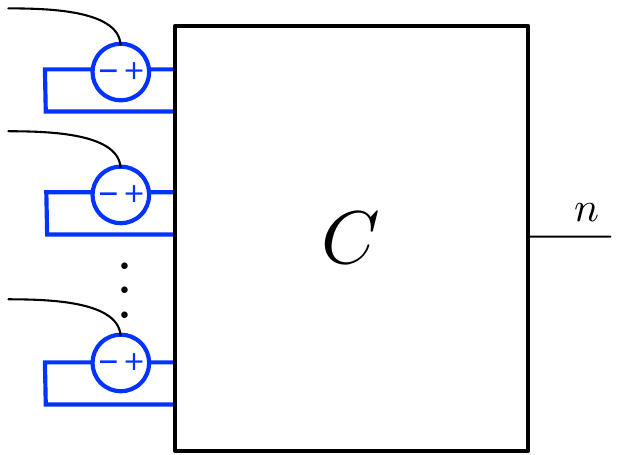}}
    \]
    The diagram illustrates this with voltage sources,
    but it applies to any mixture of sources,
    given that a turned-off current source is an open circuit.
    If moreover the outputs are a function (i.e.\ single-valued and total)
    of each source when activated individually,
    then the inclusion is an equality.
\end{theorem}

\begin{proof}
    Since $C$ has no independent sources, its mapping to $\mathsf{GAA}$ lands in the $\mathsf{GLA}$ fragment.
    Thus we can transpose (color-swap and flip) the result from the independent measurement theorem,
    noting that being a function is a self-transpose property,
    and color-swap reverses the direction of inclusion. Alternatively, one can
    use the fact that $\mathsf{GLA}$ modulo its inequational theory is an abelian category of relations~\cite{carboniCartesianBicategories1987}.
\end{proof}

The analysis of the equality case in the previous section also applies here via the duality.
Note that the theorem as stated forbids independent sources.
However it is easy to extend it to work with them too:
\[
\lower40pt\hbox{\includegraphics[height=3.2cm]{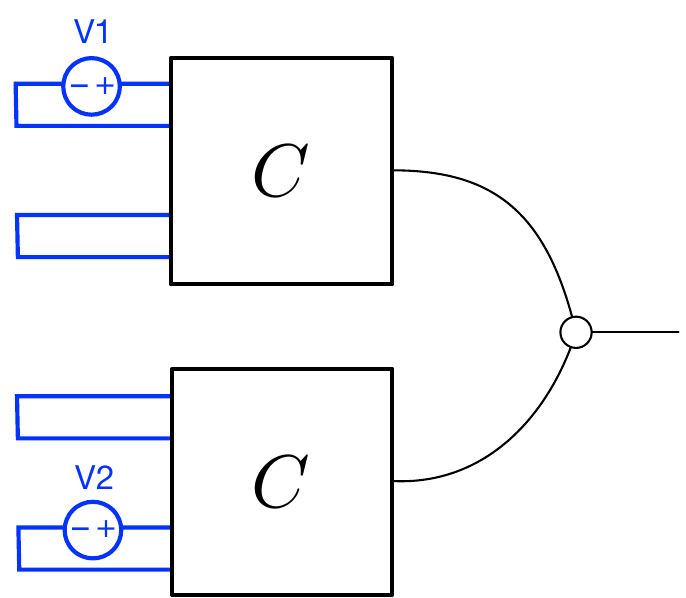}}
\ \semEq
\lower41pt\hbox{\includegraphics[height=3.3cm]{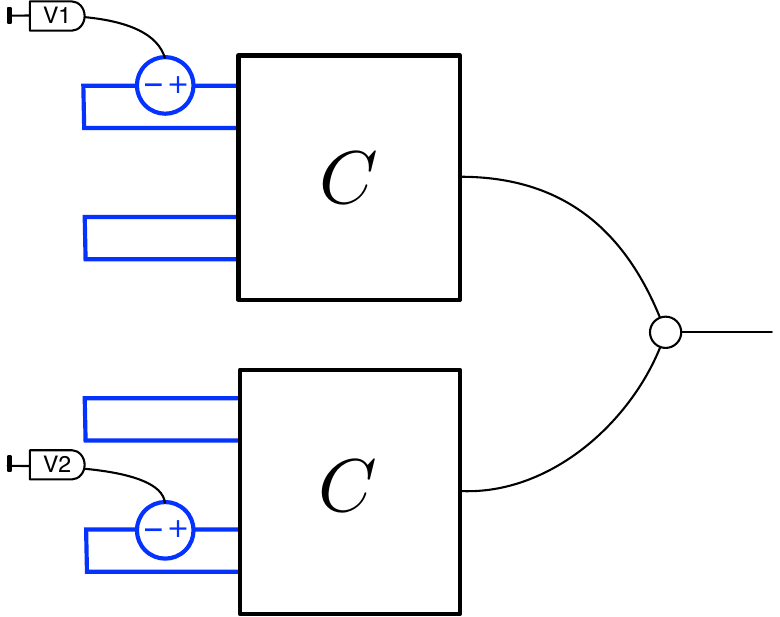}}
\semLeq\
\lower22pt\hbox{\includegraphics[height=1.7cm]{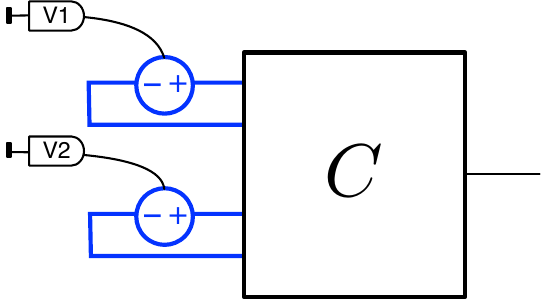}}
\semEq\
\lower22pt\hbox{\includegraphics[height=1.7cm]{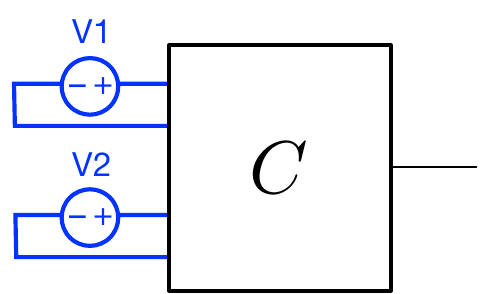}}
\]

\section{Structural theorems}

We conclude the paper with three results
that showcase the graphical language's power to state and prove non-trivial invariants. Indeed, given that our circuits are a bona fide syntax, we can use structural induction.

\subsection{Representation theorem}\label{sec:representation-thm}

Kirchhoff's laws imply two global invariants are satisfied by all circuits.
They can be elegantly stated and proved graphically,
and imply a very useful result:
all one-port circuits (those with one electrical input and one electrical output)
are representable by an impedance as described in \cref{sec:imp-calculus}.
Thus the impedance calculus can be used with any one-port circuit that we encounter.

\begin{proposition}[Relativity of potentials]\label{prop:relativePotentials}
    A circuit constrains voltage differences, not absolute voltages.
    That is, adding the same voltage difference to all open wires of a circuit does not change
    its behavior.
    \[
    \lower30pt\hbox{\includegraphics[height=3cm]{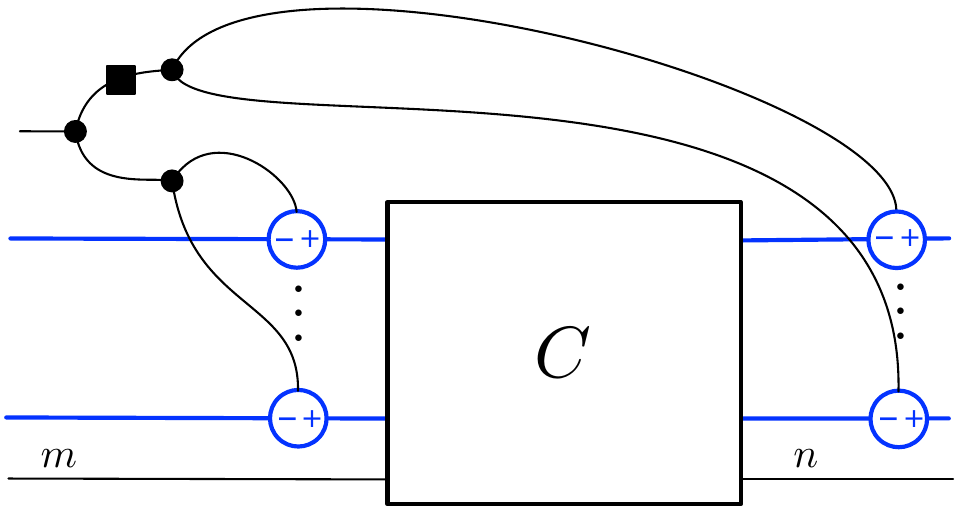}}
    \quad \semEq\quad
    \lower28pt\hbox{\includegraphics[height=2.2cm]{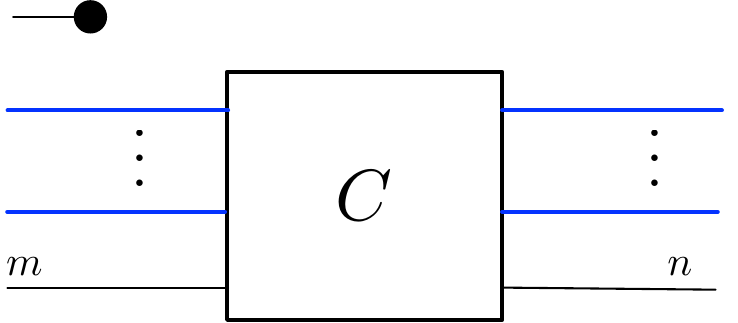}}
    \]
\end{proposition}

\begin{proof}
    First, a lemma:
    \[
    \lower12pt\hbox{\includegraphics[height=1.2cm]{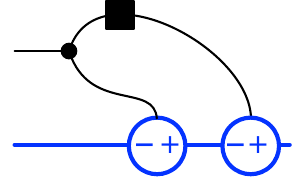}}
    \semEq
    \lower12pt\hbox{\includegraphics[height=1.2cm]{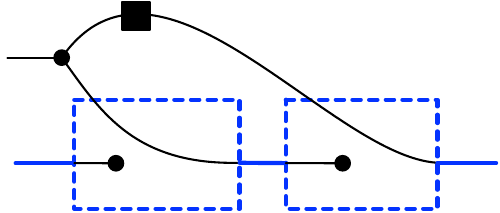}}
    \semEq
    \lower12pt\hbox{\includegraphics[height=1.2cm]{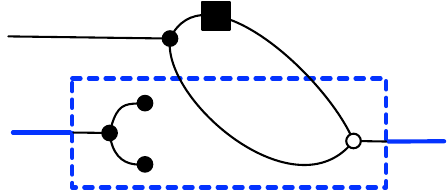}}
    \semEq
    \lower12pt\hbox{\includegraphics[height=1cm]{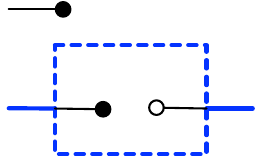}}
    \semEq\
    \lower6pt\hbox{\includegraphics[height=.5cm]{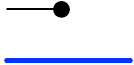}}
    \]

    Then using the impedance calculus we notice that voltage sources commute with circuit elements and junctions:
    \[
     \lower3pt\hbox{\includegraphics[height=.7cm]{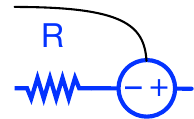}}
     \semEq
     \lower8pt\hbox{\includegraphics[height=1cm]{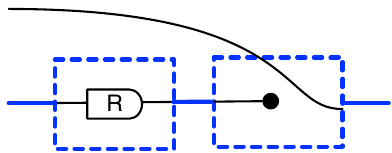}}
     \semEq
     \lower8pt\hbox{\includegraphics[height=1cm]{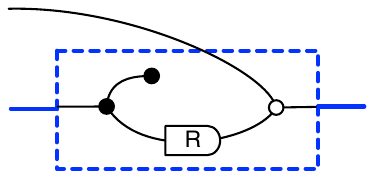}}
     \semEq
     \lower6pt\hbox{\includegraphics[height=.8cm]{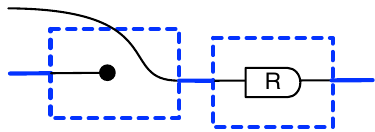}}
     \semEq
     \lower3pt\hbox{\includegraphics[height=.7cm]{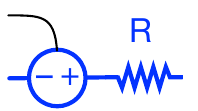}}
    \]

    \[
    \semElec{\lower10pt\hbox{\includegraphics[height=1.2cm]{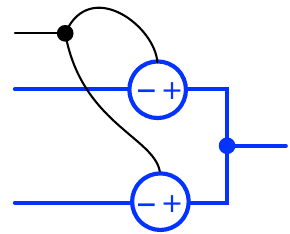}}}
    =
    \lower8pt\hbox{\includegraphics[height=1cm]{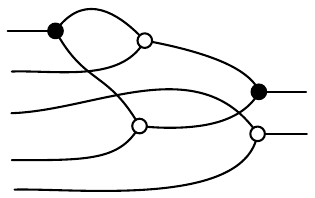}}
    =
    \lower8pt\hbox{\includegraphics[height=1cm]{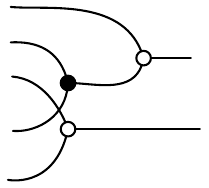}}
    =
    \semElec{\lower8pt\hbox{\includegraphics[height=1cm]{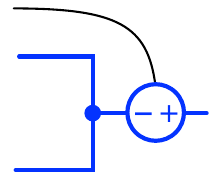}}}
    \]

    This together with the first lemma allows us to pass the voltage sources through any circuit.
    This can be made formal using induction but we omit the details for brevity.
\end{proof}

\begin{proposition}[Conservation of currents]\label{prop:conservationCurrents}
    The sum of the currents going into a circuit is equal to the sum of the outgoing currents.
    \[
    \lower30pt\hbox{\includegraphics[height=3cm]{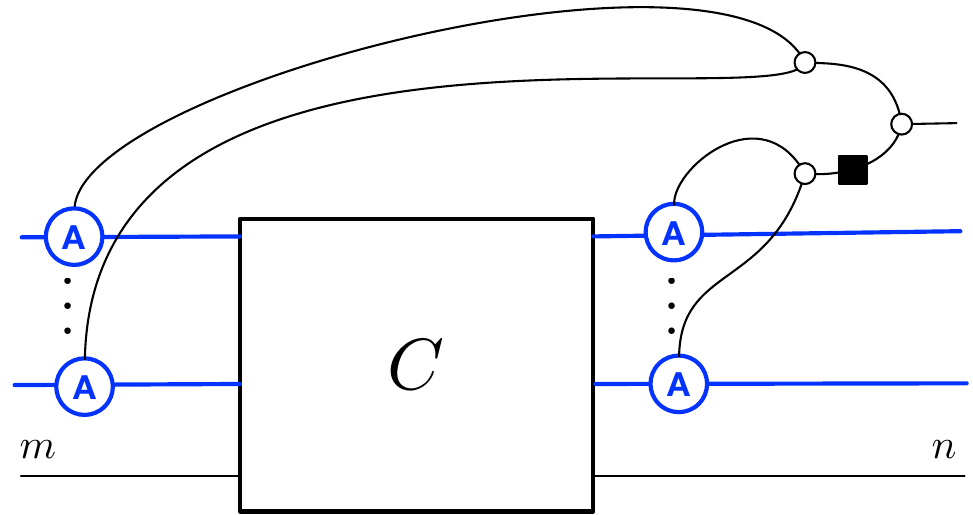}}
    \quad \semEq\quad
    \lower30pt\hbox{\includegraphics[height=2.3cm]{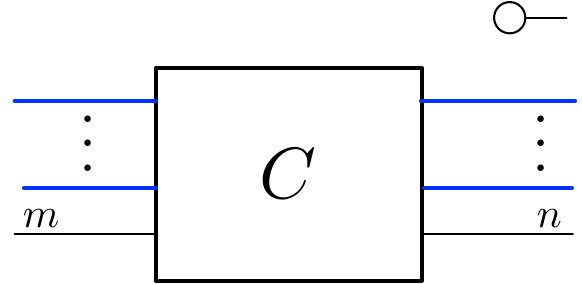}}
    \]
\end{proposition}

\begin{proof}
    The proof proceeds exactly like the previous one,
    the only difference being an inversion of the colours,
    which does not affect the proof.
\end{proof}

\begin{theorem}[Representation theorem]\label{thm:representation}
    A one-port circuit $c$
    is representable by an impedance:
    \[
    \lower5pt\hbox{\includegraphics[height=.5cm]{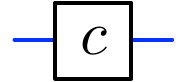}}
    \semEq
    \lower20pt\hbox{\includegraphics[height=1.6cm]{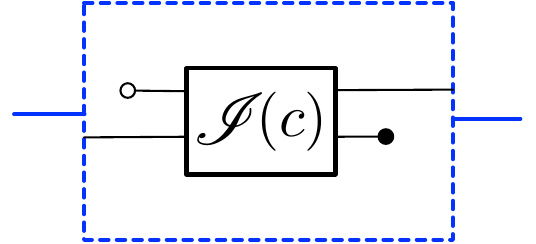}}
    \]
\end{theorem}

\begin{proof}
    Using the conclusions of Propositions~\ref{prop:relativePotentials} and~\ref{prop:conservationCurrents} we have
    \[
    \lower8pt\hbox{\includegraphics[height=.5cm]{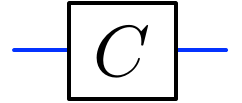}}
    \semEq
    \lower9pt\hbox{\includegraphics[height=1.2cm]{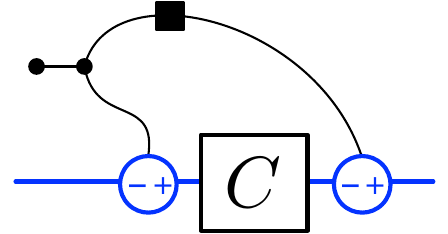}}
    \semEq
    \lower9pt\hbox{\includegraphics[height=1.5cm]{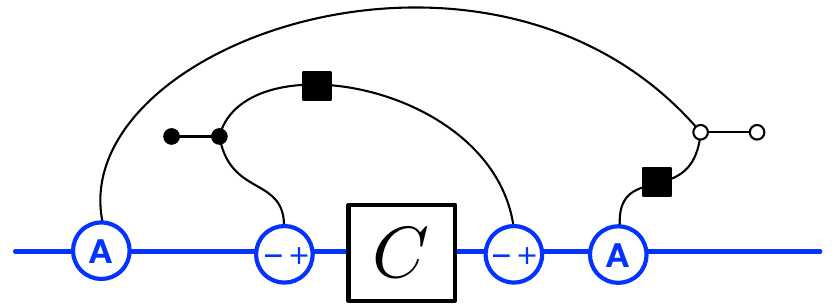}}
    \]
    A simple calculation, therefore, confirms that:
    \[
    \lower8pt\hbox{\includegraphics[height=.7cm]{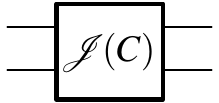}}
    =
    \lower19pt\hbox{\includegraphics[height=1.7cm]{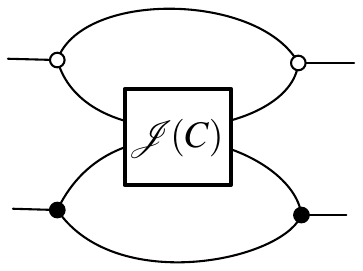}}.
    \]
    Now:
    \begin{multline*}
    \semElec{\lower15pt\hbox{\includegraphics[height=1.4cm]{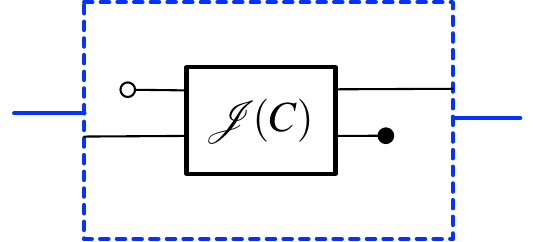}}}
    =
    \lower19pt\hbox{\includegraphics[height=1.7cm]{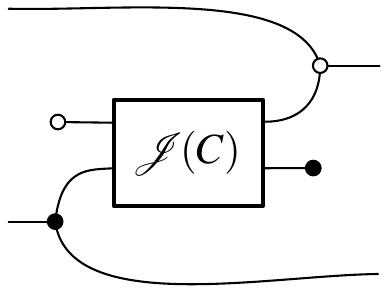}}
    =
    \lower22pt\hbox{\includegraphics[height=1.9cm]{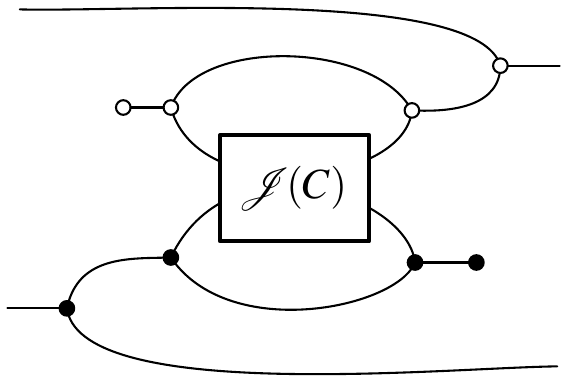}} \\
    =
    \lower30pt\hbox{\includegraphics[height=2.1cm]{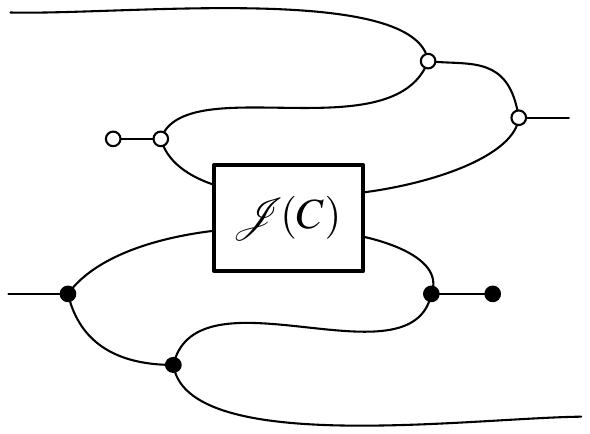}}
    =
    \lower22pt\hbox{\includegraphics[height=1.7cm]{graffles/jc2.pdf}}
    =
    \lower8pt\hbox{\includegraphics[height=.7cm]{graffles/jc1.pdf}}
    =
    \semElec{\lower4pt\hbox{\includegraphics[height=.5cm]{graffles/repcircuit1.pdf}}}.
    \end{multline*}
\end{proof}

\subsection{Expressivity}\label{sec:expressivity}

A nervous reader might ask:
given that we have liberally extended the basic calculus with
arbitrary affine relations, via the information wires,
do we risk modeling unphysical circuits?
The answer is a surprising one:
textbook elements are already expressive enough;
they can express most affine relations.

We start by observing a surprising fact:
adding \emph{only} controlled sources and meters to the basic circuit elements~\eqref{eq:circpropsig}
is \emph{already} enough to build a circuit for most affine relations.
Indeed, already we can build any $\mathsf{GAA}$ diagram on the information wires, since
every generator can be constructed with a small circuit:

\[
\Wmult \semEq \lower13pt\hbox{\includegraphics[height=1.3cm]{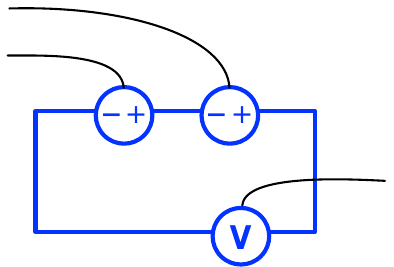}}
\quad
\Wunit \semEq \lower2pt\hbox{\includegraphics[height=.5cm]{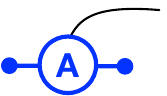}}
\quad
\Bcomult \semEq \lower13pt\hbox{\includegraphics[height=1.1cm]{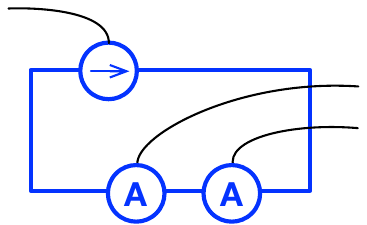}}
\quad
\Bcounit \semEq \lower2pt\hbox{\includegraphics[height=.5cm]{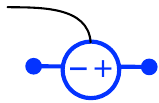}}
\]
\[
\Wcomult \semEq \lower13pt\hbox{\includegraphics[height=1.1cm]{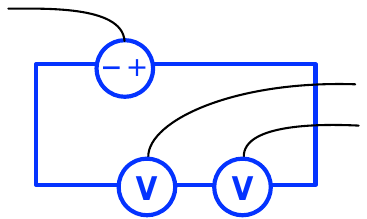}}
\quad
\Wcounit \semEq \lower2pt\hbox{\includegraphics[height=.5cm]{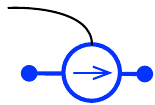}}
\quad
\Bmult \semEq \lower13pt\hbox{\includegraphics[height=1.1cm]{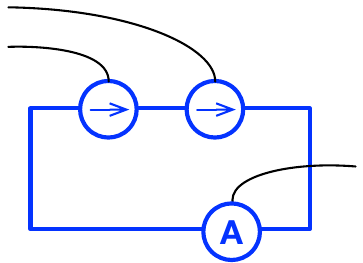}}
\quad
\Bunit \semEq \lower2pt\hbox{\includegraphics[height=.5cm]{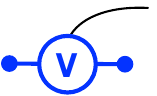}}
\]
\[
\One \semEq \lower10pt\hbox{\includegraphics[height=1.1cm]{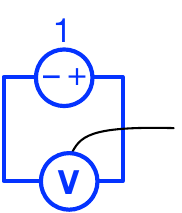}}
\quad
\antipode \semEq \lower10pt\hbox{\includegraphics[height=1.1cm]{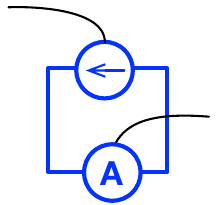}}
\quad
\lower2pt\hbox{\includegraphics[height=.25cm]{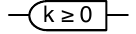}}
\semEq
\lower10pt\hbox{\includegraphics[height=1.1cm]{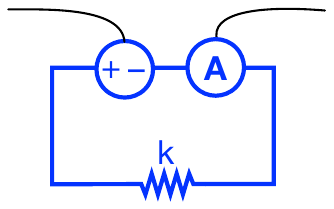}}
\quad
\circuitXop
\semEq
\lower10pt\hbox{\includegraphics[height=1.1cm]{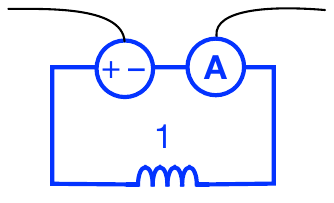}}
\quad
\circuitX
\semEq
\lower10pt\hbox{\includegraphics[height=1.1cm]{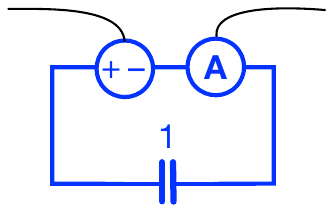}}
\]

So, restricting to just information wires, we have the full expressivity of $\mathsf{GAA}$.
The following result characterises expressivity, also taking electrical wires into account.

\begin{theorem}[Expressivity Theorem]\label{thm:expressivity}
    Using only basic circuit elements~\eqref{eq:circpropsig},
    meters and controlled sources,
    we can express all affine relations that respect the invariants of
    \cref{prop:relativePotentials,prop:conservationCurrents}.
    This includes arbitrary affine relations on the information wires.
\end{theorem}
\begin{proof}
    We sketched the main idea above; details omitted for space.
\end{proof}

Now take a circuit with \emph{no open information wires}
made up of only basic elements, meters and sources.
It is important to note that, although we can construct circuits equivalent to an arbitrary affine relations on the information wires, there is no \emph{direct} way to copy information wires in a circuit.
Since they cannot be copied,
every meter is connected to exactly one source. 
This is essentially equivalent to adding textbook elements called
\textit{measurement-controlled sources}.
The voltage-controlled current source is constructed as follows;
other combinations (e.g.\ current-controlled voltage source) are similar.
\[
\lower12pt\hbox{\includegraphics[height=1.2cm]{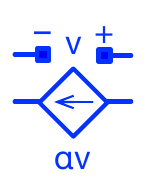}}
\quad \Defeq \quad
\lower13pt\hbox{\includegraphics[height=1.3cm]{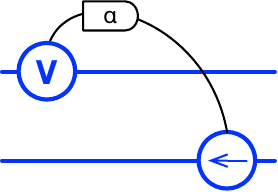}}
\]

But this also means that,
using only basic elements and measurement-controlled sources,
we can build a circuit for any relation that respects the invariants
(with no open information wires).
There is no need to fear, therefore:
textbook elements already capture enough affine relations.

\subsection{Thévenin's theorem}

Last but not least,
we prove Thévenin's theorem, 
a well-known real-life example of compositional, diagrammatic reasoning. 
It allows one to replace a one-port circuit by a simpler, equivalent one.
It can be seen both as a consequence and a stronger version of the representation theorem
(\Cref{thm:representation}).

\begin{theorem}[Thévenin's theorem]\label{thm:theveninFull}
    If $C$ is a one-port circuit made only of resistors and independent sources,
    then one of the following is true:
    \begin{enumerate}[(i)]
    \item $\lower5pt\hbox{\includegraphics[height=.5cm]{graffles/oneportcircuit.pdf}} \semEq \lower4pt\hbox{\includegraphics[height=.8cm]{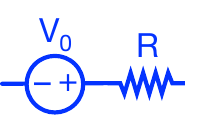}}$ for some $V_0$ and $R$,
    \item $\lower5pt\hbox{\includegraphics[height=.5cm]{graffles/oneportcircuit.pdf}} \semEq
    \lower4pt\hbox{\includegraphics[height=.8cm]{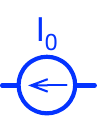}}$ for some $I_0$,
    \item $\lower5pt\hbox{\includegraphics[height=.5cm]{graffles/oneportcircuit.pdf}}$ denotes the empty relation.
    \end{enumerate}
\end{theorem}
\begin{proof}
    Omitted for space.
\end{proof}

\section{Conclusion and future work}

We extended existing compositional treatments of electrical circuit theory by including meters and controlled sources, enabling the analysis of closed circuits. We introduced the impedance calculus and demonstrated its power by solving simple circuits, and by proving a number of results---with the superposition theorem and Th\'evenin's theorem particularly notable, well-known examples from the corpus of standard electrical circuit theory. We also proved a representation theorem and characterised the expressivity of our circuit language, which is exactly that of circuits built up of standard textbook elements. With the syntactic approach, proofs are pleasingly simple, involving simple invariants and induction.

\medskip
It is worth stepping back to consider the bigger picture. We believe that the results herein are evidence that, while it is uncontroversial to say that electrical circuits are relational entities, they are better understood using relational \emph{mathematics}. The benefit of not turning everything into matrices is that theorems are more crisp to state, the relationships and symmetries between various results become more immediately apparent, annoying ``degenerate'' cases are not ignored, and proofs become more satisfying.

There is much future work: while we could not include all of the theorems, we are confident of being able to replicate the majority of standard electrical circuit theory using these techniques. An obvious open question is how far can this approach take us? The passage from Graphical Linear Algebra to Graphical Affine Algebra allows to pass from passive circuits to non-passive circuits. The next, obvious challenge is to go beyond linearity: for example, if we were  to satisfactorily capture the diode, we would get transistors, enabling the compositional analysis of arbitrary electronic circuits.

\bibliographystyle{eptcs}
\bibliography{main}

\end{document}

%% file: tikz/generators/zero.tikz
\begin{tikzpicture}[baseline=-.5ex, scale=.7]
	\begin{pgfonlayer}{nodelayer}
		\node [style=white] (0) at (0.25, 0) {};
		\node [style=none] (1) at (1.75, 0) {};
	\end{pgfonlayer}
	\begin{pgfonlayer}{edgelayer}
		\draw (0) to (1.center);
	\end{pgfonlayer}
\end{tikzpicture}

%% file: tikz/generators/delete.tikz
\begin{tikzpicture}[baseline=-.5ex, scale=.7]
	\begin{pgfonlayer}{nodelayer}
		\node [style=black] (0) at (1.5, 0) {};
		\node [style=none] (1) at (0, 0) {};
	\end{pgfonlayer}
	\begin{pgfonlayer}{edgelayer}
		\draw (0) to (1.center);
	\end{pgfonlayer}
\end{tikzpicture}

%% file: tikz/generators/copy.tikz
\begin{tikzpicture}
	\begin{pgfonlayer}{nodelayer}
		\node [style=black] (0) at (0, 0) {};
		\node [style=none] (1) at (-1, 0) {};
		\node [style=none] (2) at (1, -0.5) {};
		\node [style=none] (3) at (1, 0.5) {};
	\end{pgfonlayer}
	\begin{pgfonlayer}{edgelayer}
		\draw (0) to (1.center);
		\draw [in=-60, out=180] (2.center) to (0);
		\draw [in=60, out=-180] (3.center) to (0);
	\end{pgfonlayer}
\end{tikzpicture}

%% file: tikz/generators/add.tikz
\begin{tikzpicture}
	\begin{pgfonlayer}{nodelayer}
		\node [style=white] (0) at (0, 0) {};
		\node [style=none] (1) at (1, 0) {};
		\node [style=none] (2) at (-1, -0.5) {};
		\node [style=none] (3) at (-1, 0.5) {};
	\end{pgfonlayer}
	\begin{pgfonlayer}{edgelayer}
		\draw (0) to (1.center);
		\draw [in=-120, out=0] (2.center) to (0);
		\draw [in=120, out=0] (3.center) to (0);
	\end{pgfonlayer}
\end{tikzpicture}

%% file: tikz/generators/scalar.tikz
\begin{tikzpicture}
	\begin{pgfonlayer}{nodelayer}
		\node [style=reg] (0) at (-0.25, 0) {$k$};
		\node [style=none] (1) at (1, 0) {};
		\node [style=none] (2) at (-1.5, 0) {};
	\end{pgfonlayer}
	\begin{pgfonlayer}{edgelayer}
		\draw (2.center) to (0);
		\draw (0) to (1.center);
	\end{pgfonlayer}
\end{tikzpicture}

%% file: tikz/generators/register.tikz
\begin{tikzpicture}
	\begin{pgfonlayer}{nodelayer}
		\node [style=reg] (0) at (-0.25, -0) {$x$};
		\node [style=none] (1) at (1, -0) {};
		\node [style=none] (2) at (-1.5, -0) {};
	\end{pgfonlayer}
	\begin{pgfonlayer}{edgelayer}
		\draw (2.center) to (0);
		\draw (0) to (1.center);
	\end{pgfonlayer}
\end{tikzpicture}

%% file: tikz/generators/co-delete.tikz
\begin{tikzpicture}[baseline=-.5ex, scale=.7]
	\begin{pgfonlayer}{nodelayer}
		\node [style=black] (0) at (0.25, 0) {};
		\node [style=none] (1) at (1.75, 0) {};
	\end{pgfonlayer}
	\begin{pgfonlayer}{edgelayer}
		\draw (0) to (1.center);
	\end{pgfonlayer}
\end{tikzpicture}

%% file: tikz/generators/co-zero.tikz
\begin{tikzpicture}
	\begin{pgfonlayer}{nodelayer}
		\node [style=white] (0) at (1.5, 0) {};
		\node [style=none] (1) at (0.5, 0) {};
	\end{pgfonlayer}
	\begin{pgfonlayer}{edgelayer}
		\draw (0) to (1.center);
	\end{pgfonlayer}
\end{tikzpicture}

%% file: tikz/generators/co-add.tikz
\begin{tikzpicture}
	\begin{pgfonlayer}{nodelayer}
		\node [style=white] (0) at (0, 0) {};
		\node [style=none] (1) at (-1, 0) {};
		\node [style=none] (2) at (1, -0.5) {};
		\node [style=none] (3) at (1, 0.5) {};
	\end{pgfonlayer}
	\begin{pgfonlayer}{edgelayer}
		\draw (0) to (1.center);
		\draw [in=-60, out=180] (2.center) to (0);
		\draw [in=60, out=180] (3.center) to (0);
	\end{pgfonlayer}
\end{tikzpicture}

%% file: tikz/generators/co-copy.tikz
\begin{tikzpicture}
	\begin{pgfonlayer}{nodelayer}
		\node [style=black] (0) at (0, 0) {};
		\node [style=none] (1) at (1, 0) {};
		\node [style=none] (2) at (-1, -0.5) {};
		\node [style=none] (3) at (-1, 0.5) {};
	\end{pgfonlayer}
	\begin{pgfonlayer}{edgelayer}
		\draw (0) to (1.center);
		\draw [in=-120, out=0] (2.center) to (0);
		\draw [in=120, out=0] (3.center) to (0);
	\end{pgfonlayer}
\end{tikzpicture}

%% file: tikz/generators/co-scalar.tikz
\begin{tikzpicture}
	\begin{pgfonlayer}{nodelayer}
		\node [style=coreg] (0) at (-0.25, 0) {$k$};
		\node [style=none] (1) at (1, 0) {};
		\node [style=none] (2) at (-1.5, 0) {};
	\end{pgfonlayer}
	\begin{pgfonlayer}{edgelayer}
		\draw (2.center) to (0);
		\draw (0) to (1.center);
	\end{pgfonlayer}
\end{tikzpicture}

%% file: tikz/generators/co-register.tikz
\begin{tikzpicture}
	\begin{pgfonlayer}{nodelayer}
		\node [style=coreg] (0) at (-0.25, 0) {$x$};
		\node [style=none] (1) at (1, 0) {};
		\node [style=none] (2) at (-1.5, 0) {};
	\end{pgfonlayer}
	\begin{pgfonlayer}{edgelayer}
		\draw (2.center) to (0);
		\draw (0) to (1.center);
	\end{pgfonlayer}
\end{tikzpicture}

%% file: tikz/one.tikz
\begin{tikzpicture}[baseline=-.5ex]
	\begin{pgfonlayer}{nodelayer}
		\node [style=none] (0) at (-0.75, -0) {};
		\node [style=none] (1) at (0.25, -0) {};
		\node [style=none] (2) at (-0.75, 0.25) {};
		\node [style=none] (3) at (-0.75, -0.25) {};
	\end{pgfonlayer}
	\begin{pgfonlayer}{edgelayer}
		\draw (0.center) to (1.center);
		\draw (3.center) to (2.center);
	\end{pgfonlayer}
\end{tikzpicture}

%% file: tikz/one-copy.tikz
\begin{tikzpicture}
	\begin{pgfonlayer}{nodelayer}
		\node [style=none] (0) at (-1.25, -0) {};
		\node [style=none] (1) at (1.25, 0.75) {};
		\node [style=black] (2) at (0, -0) {};
		\node [style=none] (3) at (1.25, -0.75) {};
		\node [style=none] (4) at (-1.25, -0.25) {};
		\node [style=none] (5) at (-1.25, 0.25) {};
	\end{pgfonlayer}
	\begin{pgfonlayer}{edgelayer}
		\draw [bend right, looseness=1.00] (1.center) to (2);
		\draw [in=180, out=-60, looseness=1.00] (2) to (3.center);
		\draw (0.center) to (2);
		\draw (5.center) to (4.center);
	\end{pgfonlayer}
\end{tikzpicture}

%% file: tikz/one2.tikz
\begin{tikzpicture}
	\begin{pgfonlayer}{nodelayer}
		\node [style=none] (0) at (0.75, -0.5) {};
		\node [style=none] (1) at (-0.75, 0.75) {};
		\node [style=none] (2) at (-0.75, -0.25) {};
		\node [style=none] (3) at (-0.75, -0.75) {};
		\node [style=none] (4) at (0.75, 0.5) {};
		\node [style=none] (5) at (-0.75, 0.5) {};
		\node [style=none] (6) at (-0.75, -0.5) {};
		\node [style=none] (7) at (-0.75, 0.25) {};
	\end{pgfonlayer}
	\begin{pgfonlayer}{edgelayer}
		\draw (5.center) to (4.center);
		\draw (1.center) to (7.center);
		\draw (6.center) to (0.center);
		\draw (2.center) to (3.center);
	\end{pgfonlayer}
\end{tikzpicture}

%% file: tikz/one-delete.tikz
\begin{tikzpicture}
	\begin{pgfonlayer}{nodelayer}
		\node [style=none] (0) at (-0.75, 0) {};
		\node [style=black] (1) at (0.5, 0) {};
		\node [style=none] (2) at (-0.75, -0.25) {};
		\node [style=none] (3) at (-0.75, 0.25) {};
	\end{pgfonlayer}
	\begin{pgfonlayer}{edgelayer}
		\draw (0.center) to (1);
		\draw (3.center) to (2.center);
	\end{pgfonlayer}
\end{tikzpicture}

%% file: tikz/empty-diag.tikz
\begin{tikzpicture}
	\begin{pgfonlayer}{nodelayer}
		\node [style=none] (0) at (-0.5, 0.75) {};
		\node [style=none] (1) at (0.75, 0.75) {};
		\node [style=none] (2) at (-0.5, -0.5) {};
		\node [style=none] (3) at (0.75, -0.5) {};
	\end{pgfonlayer}
	\begin{pgfonlayer}{edgelayer}
		\draw [densely dotted] (0.center) to (1.center);
		\draw [densely dotted] (1.center) to (3.center);
		\draw [densely dotted] (3.center) to (2.center);
		\draw [densely dotted] (2.center) to (0.center);
	\end{pgfonlayer}
\end{tikzpicture}

%% file: tikz/one-false.tikz
\begin{tikzpicture}
	\begin{pgfonlayer}{nodelayer}
		\node [style=none] (0) at (-0.5, 0.5) {};
		\node [style=white] (1) at (0.5, 0.5) {};
		\node [style=none] (2) at (-0.5, 0.25) {};
		\node [style=none] (3) at (-0.5, 0.75) {};
		\node [style=none] (4) at (-2, -0.5) {};
		\node [style=none] (5) at (1.75, -0.5) {};
	\end{pgfonlayer}
	\begin{pgfonlayer}{edgelayer}
		\draw (0.center) to (1);
		\draw (3.center) to (2.center);
		\draw (4.center) to (5.center);
	\end{pgfonlayer}
\end{tikzpicture}

%% file: tikz/one-false-disconnect.tikz
\begin{tikzpicture}
	\begin{pgfonlayer}{nodelayer}
		\node [style=none] (0) at (-0.5, 0.5) {};
		\node [style=none] (1) at (2, -0.5) {};
		\node [style=none] (2) at (-2, -0.5) {};
		\node [style=black] (3) at (-0.5, -0.5) {};
		\node [style=white] (4) at (0.5, 0.5) {};
		\node [style=black] (5) at (0.5, -0.5) {};
		\node [style=none] (6) at (-0.5, 0.25) {};
		\node [style=none] (7) at (-0.5, 0.75) {};
	\end{pgfonlayer}
	\begin{pgfonlayer}{edgelayer}
		\draw (0.center) to (4);
		\draw (7.center) to (6.center);
		\draw (2.center) to (3);
		\draw (1.center) to (5);
	\end{pgfonlayer}
\end{tikzpicture}

%% file: main.bbl
\begin{thebibliography}{1}
\providecommand{\bibitemdeclare}[2]{}
\providecommand{\surnamestart}{}
\providecommand{\surnameend}{}
\providecommand{\urlprefix}{Available at }
\providecommand{\url}[1]{\texttt{#1}}
\providecommand{\href}[2]{\texttt{#2}}
\providecommand{\urlalt}[2]{\href{#1}{#2}}
\providecommand{\doi}[1]{doi:\urlalt{http://dx.doi.org/#1}{#1}}
\providecommand{\bibinfo}[2]{#2}

\bibitemdeclare{article}{baezCompositionalFrameworkPassive2015}
\bibitem{baezCompositionalFrameworkPassive2015}
\bibinfo{author}{John~C. \surnamestart Baez\surnameend} \&
  \bibinfo{author}{Brendan \surnamestart Fong\surnameend}
  (\bibinfo{year}{2018}): \emph{\bibinfo{title}{A {{Compositional Framework}}
  for {{Passive Linear Networks}}}}.
\newblock {\sl \bibinfo{journal}{Theory and Applications of Categories}}
  \bibinfo{volume}{33}, pp. \bibinfo{pages}{1158--1222}.
\newblock
  \urlprefix\url{http://www.tac.mta.ca/tac/volumes/33/38/33-38abs.html}.

\bibitemdeclare{inproceedings}{bonchiRefinementSignalFlow2017}
\bibitem{bonchiRefinementSignalFlow2017}
\bibinfo{author}{Filippo \surnamestart Bonchi\surnameend},
  \bibinfo{author}{Joshua \surnamestart Holland\surnameend},
  \bibinfo{author}{Dusko \surnamestart Pavlovic\surnameend} \&
  \bibinfo{author}{Pawe{\l} \surnamestart Soboci{\'n}ski\surnameend}
  (\bibinfo{year}{2017}): \emph{\bibinfo{title}{Refinement for Signal Flow
  Graphs}}.
\newblock In \bibinfo{editor}{Roland \surnamestart Meyer\surnameend} \&
  \bibinfo{editor}{Uwe \surnamestart Nestmann\surnameend}, editors: {\sl
  \bibinfo{booktitle}{28th International Conference on Concurrency Theory
  ({{CONCUR}} 2017)}}, {\sl \bibinfo{series}{Leibniz International Proceedings
  in Informatics ({{LIPIcs}})}}~\bibinfo{volume}{85},
  \bibinfo{publisher}{{Schloss Dagstuhl\textendash Leibniz-Zentrum fuer
  Informatik}}, \bibinfo{address}{{Dagstuhl, Germany}}, pp.
  \bibinfo{pages}{24:1--24:16}, \doi{10.4230/LIPIcs.CONCUR.2017.24}.
\newblock \urlprefix\url{http://drops.dagstuhl.de/opus/volltexte/2017/7775}.

\bibitemdeclare{inproceedings}{bonchiGraphicalAffineAlgebra2019}
\bibitem{bonchiGraphicalAffineAlgebra2019}
\bibinfo{author}{Filippo \surnamestart Bonchi\surnameend},
  \bibinfo{author}{Robin \surnamestart Piedeleu\surnameend},
  \bibinfo{author}{Pawe{\l} \surnamestart Soboci{\'n}ski\surnameend} \&
  \bibinfo{author}{Fabio \surnamestart Zanasi\surnameend}
  (\bibinfo{year}{2019}): \emph{\bibinfo{title}{Graphical {{Affine Algebra}}}}.
\newblock In: {\sl \bibinfo{booktitle}{34th {{Annual ACM}}/{{IEEE Symposium}}
  on {{Logic}} in {{Computer Science}} ({{LICS}})}},
  \bibinfo{publisher}{{IEEE}}, \bibinfo{address}{{Vancouver, BC, Canada}}, pp.
  \bibinfo{pages}{1--12}, \doi{10.1109/LICS.2019.8785877}.

\bibitemdeclare{article}{bonchiInteractingHopfAlgebras2017}
\bibitem{bonchiInteractingHopfAlgebras2017}
\bibinfo{author}{Filippo \surnamestart Bonchi\surnameend},
  \bibinfo{author}{Pawe{\l} \surnamestart Soboci{\'n}ski\surnameend} \&
  \bibinfo{author}{Fabio \surnamestart Zanasi\surnameend}
  (\bibinfo{year}{2017}): \emph{\bibinfo{title}{Interacting {{Hopf
  Algebras}}}}.
\newblock {\sl \bibinfo{journal}{Journal of Pure and Applied Algebra}}
  \bibinfo{volume}{221}(\bibinfo{number}{1}), pp. \bibinfo{pages}{144--184},
  \doi{10.1016/j.jpaa.2016.06.002}.

\bibitemdeclare{article}{carboniCartesianBicategories1987}
\bibitem{carboniCartesianBicategories1987}
\bibinfo{author}{Aurelio \surnamestart Carboni\surnameend} \&
  \bibinfo{author}{R.~F.~C. \surnamestart Walters\surnameend}
  (\bibinfo{year}{1987}): \emph{\bibinfo{title}{Cartesian Bicategories {{I}}}}.
\newblock {\sl \bibinfo{journal}{Journal of Pure and Applied Algebra}}
  \bibinfo{volume}{49}, pp. \bibinfo{pages}{11--32},
  \doi{10.1016/0022-4049(87)90121-6}.

\bibitemdeclare{phdthesis}{coyaCircuitsBondGraphs2018}
\bibitem{coyaCircuitsBondGraphs2018}
\bibinfo{author}{Brandon \surnamestart Coya\surnameend} (\bibinfo{year}{2018}):
  \emph{\bibinfo{title}{Circuits, Bond Graphs, and Signal-Flow Diagrams: {{A}}
  Categorical Perspective}}.
\newblock Ph.D. thesis, \bibinfo{school}{University of California Riverside}.

\bibitemdeclare{book}{desoerBasicCircuitTheory1969}
\bibitem{desoerBasicCircuitTheory1969}
\bibinfo{author}{Charles~A. \surnamestart Desoer\surnameend} \&
  \bibinfo{author}{Ernest~S. \surnamestart Kuh\surnameend}
  (\bibinfo{year}{1969}): \emph{\bibinfo{title}{Basic Circuit Theory}}.
\newblock \bibinfo{publisher}{{McGraw-Hill}}, \bibinfo{address}{{New York}}.

\end{thebibliography}
